%% file: dilations.tex
\documentclass[11pt,a4paper]{article}

\usepackage{hyperref}
\usepackage{a4wide}
\usepackage{amsmath,amsthm,amsxtra}
\usepackage[latin1]{inputenc}
\usepackage[all]{xy}

\newtheorem{definition}{Definition}
\newtheorem{theorem}[definition]{Theorem}
\newtheorem{lemma}[definition]{Lemma}
\newtheorem{proposition}[definition]{Proposition}
\newtheorem{corollary}[definition]{Corollary}

\numberwithin{equation}{section}
\numberwithin{definition}{section}

\input{macros}

\input{scaling-macros}

\begin{document}

\title{On dilation symmetries arising from scaling limits}
\author{Henning Bostelmann\thanks{%
Supported in parts by the EU network ``Noncommutative Geometry''
(MRTN-CT-2006-0031962)} 
\and Claudio D'Antoni$^\ast$\thanks{%
Supported in parts by PRIN-MIUR and GNAMPA-INDAM}
\and Gerardo Morsella$^{\ast\dagger}$\thanks{%
Supported in parts by the Scuola Normale Superiore, Pisa, Italy}
} 
\date{Università di Roma ``Tor Vergata'', Dipartimento di Matematica,\\ 
Via della Ricerca Scientifica, 00133 Roma, Italy
\\
\medskip
E-mail addresses: ({\tt{}bostelma,dantoni,morsella}){\tt{}@mat.uniroma2.it}
\\
\medskip
December 27, 2008}

\maketitle

\begin{abstract}
Quantum field theories, at short scales, can be approximated by a
scaling limit theory. In this approximation, an additional
symmetry is gained, namely dilation covariance. To understand the
structure of this dilation symmetry, we investigate it in a
nonperturbative, model independent context. To that end, it turns
out to be necessary to consider non-pure vacuum states in the
limit. These can be decomposed into an integral of pure states; we
investigate how the symmetries and observables of the theory
behave under this decomposition. In particular, we consider
several natural conditions of increasing strength that yield
restrictions on the decomposed dilation symmetry.
\end{abstract}

\section{Introduction} \label{sec:intro}

In the analysis of quantum field theories, the notion of scaling limits plays
an important role. The physical picture underlying this mathematical concept is
as follows: One considers measurements in smaller and smaller space-time
regions, at the same time increasing the energy content of the states
involved, so that the characteristic action scale remains constant. Passing to
the limit of infinitesimal scales, one obtains a new quantum field theory, the
\emph{scaling limit} of the original model. 
The scaling limit theory can be
seen as an approximation of the full theory in the short-distance regime.
However, it may differ significantly from the full theory in fundamental aspects, for example regarding its charge structure: In quantum
chromodynamics, it is expected that confined charges (color) appear in the
limit theory, but are not visible as such in the full theory.

The virtues of the scaling limit theory include that it is typically
\emph{simpler} than the original one. In fact, in relevant examples, one
expects it to be interactionless (asymptotic freedom). But even where this is
not the case, the limit theory should possess an additional symmetry: It should
be dilation covariant, since any finite masses in the original model can be
neglected in the limit of large energies.

On the mathematical side, a very natural description of scaling limits has been
given by Buchholz and Verch \cite{BucVer:scaling_algebras}. This description,
formulated in the \cistar{} algebraic framework of local quantum physics
\cite{Haa:LQP}, originates directly from the physical notions, and avoids any
additional input motivated merely on the technical side, such as a rescaling of
coupling constants or mass parameters, or the choice of renormalization factors
for quantum fields.  This has the advantage of allowing an intrinsic, model-independent
description of the short distance properties of the theory at hand. In particular,
it has been successfully applied to the analysis of the charge structure
of the theory in the scaling limit and to the intrinsic characterization of
charge confinement~\cite{Buc:quarks, DMV:scaling_charges}.
While the framework of Buchholz and Verch seems rather
abstract at first, it has recently been shown that it reproduces the usual
picture of multiplicative field renormalization in typical cases
\cite{BDM:scaling_fields}. 
 
The approach of \cite{BucVer:scaling_algebras} is based on the notion of the
\emph{scaling algebra} $\uafk$, which consists -- roughly speaking -- of
sequences of observables $\lambda \mapsto \uA_\lambda$ at varying scale
$\lambda$, uniformly bounded in norm, and subject to certain continuity conditions. (We shall recall
the precise definition in Sec.~\ref{sec:setting}.) The task of passing to the
scaling limit is then reduced to finding a suitable state $\uomega_0$ on
the \cistar{} algebra $\uafk$ that represents the vacuum of the limit theory; it is
constructed as a limit of vacuum states at finite scales. The limit theory
itself is then obtained by a standard GNS construction with respect to
$\uomega_0$.

It should be easy in this context to describe the additional dilation symmetry
that arises in the scaling limit. In fact, the scaling algebra $\uafk$ carries
a very natural representation $\mu \mapsto \udelta_\mu$ of the dilation group, 
which acts by shifting the argument of the 
functions $\lambda \mapsto \uA_\lambda$: $\udelta_\mu(\uA)_\lambda = \uA_{\mu\lambda}$. 
However, things turn out to be more
involved: The limit states $\uomega_0$ described in
\cite{BucVer:scaling_algebras} are not invariant under this group action, and
thus one does not obtain a canonical group representation in the
limit Hilbert space. In
\cite{BDM:scaling_fields}, generalized limit states have been introduced, some
of which are invariant under dilations, and give rise to a unitary
implementation of the dilation group in the limit theory. However, these
dilation invariant limit states are never pure; rather they arise as a mixture
of states of the Buchholz-Verch type, which are pure in 2+1 or more space-time
dimensions.

The object of the present paper is to analyze this generalized class of limit
states in more detail, in order to describe the structure of the dilation symmetry associated
with the dilation invariant ones. In particular we will show that, as briefly mentioned
in \cite{BDM:scaling_fields}, the decomposition of these states in pure
(Buchholz-Verch type) states gives rise to a direct
integral decomposition of the limit Hilbert space, which also induces a
decomposition of observables and of Poincar\'e symmetries. It should be noted
here that the entire construction is complicated by the fact that uncountably many
extremal states are involved in this decomposition, and that the measure space
underlying the direct integral is of a very general nature. Because of this,
we need to use a notion of direct integral of Hilbert spaces which is more general
than the one previously employed in the quantum field theory literature
\cite{DriSum:central_decomp}.

It is also of interest to discuss how the special but physically important class of
theories with a \emph{unique scaling limit,} as defined in
\cite{BucVer:scaling_algebras}, fits into our generalized framework. It turns out
that, up to some technical conditions, uniqueness of the scaling limit
in the Buchholz-Verch framework is equivalent to the factorization of our
generalized scaling limit into a tensor product of an irreducible scaling limit
theory and a commutative part, which is just the image under the scaling
limit representation of the center of the scaling algebra. In particular
we show that such factorization holds for a restricted class of theories, those
with a \emph{convergent scaling limit}. This class includes in particular
dilation invariant theories and free field models. The technical conditions
referred to above consist in a suitable separability requirement of the scaling
limit Hilbert space, which is needed in order to be able to employ the full power 
of direct integrals theory. As a matter of fact, such separability condition is
a consequence of a refined version of the Haag-Swieca compactness condition.

With these results at hand, it is possible to discuss the structure of the unitarily
implemented dilation symmetry in dilation invariant scaling limit states. The
outcome is that in general the dilations do not decompose, not even in
the factorizing situation. Rather, the dilations intertwine in a suitable sense the different
pure limit states that occur in the direct integral decomposition. A complete 
factorization of the dilation symmetry is however obtained in the convergent
scaling limit case. For such theories, therefore, one gets a unitary implementation 
of the dilation symmetry in the pure limit theory.

The remainder of this paper is organized as follows: 
First, in
Sec.~\ref{sec:general}, we recall the notion of scaling limits in the algebraic
approach to quantum field theory, and generalize some fundamental results of
\cite{BucVer:scaling_algebras} to our situation. In
Sec.~\ref{sec:decompTheory}, we establish the direct integral decomposition
mentioned above, including a decomposition of local observables and Poincaré
symmetries.
Sec.~\ref{sec:uniqueFactorLim} contains a discussion of unique scaling limits
as a special case. We define several conditions that generalize the notion from
\cite{BucVer:scaling_algebras}, and discuss relations between them. Then, in
Sec.~\ref{sec:dilations}, we analyze the structure of dilation symmetries in
the limit Hilbert space, and their decomposition along the direct integral, on
different levels of generality.
In Sec.~\ref{sec:compactness} we propose a stronger version of the Haag-Swieca
compactness condition and we show that it implies the separability property
used in the analysis of Sec.~\ref{sec:uniqueFactorLim}.
Sec.~\ref{sec:examples} discusses some simple models as examples, showing in
particular that these fulfill all of our conditions proposed in
Sec.~\ref{sec:uniqueFactorLim} and~\ref{sec:compactness}. We conclude with a brief outlook in
Sec.~\ref{sec:conclusion}. The appendix reviews the concept of direct
integrals of Hilbert spaces, which we need in a more general variant than
covered in the standard literature.

\section{Definitions and general results} \label{sec:general}

We shall first recall the definition of scaling limits in the algebraic
approach to quantum field theory, and prove some fundamental results regarding
uniqueness of the limit vacuum state and regarding geometric modular action.

\subsection{The setting} \label{sec:setting}

We consider quantum field theory on $(s+1)$ dimensional Minkowski space. For our
analysis, we work entirely within the framework of algebraic quantum field theory \cite{Haa:LQP}, where observables localized in a space-time region
$\ocal$ are described by the selfadjoint elements of a \cistar{} algebra
$\afk(\ocal)$. Let us repeat the formal definition of a quantum field
theoretical model in this context.

\begin{definition} \label{def:localnet}
Let $\gcal$ be a Lie group of point transformations of Minkowski space that includes the 
translation group. A \emph{local net of algebras} with symmetry group $\gcal$
is a net of algebras $\afk$ together with a representation $g \mapsto \alpha_g$ of $\gcal$ 
as automorphisms of $\afk$, such that
\begin{enumerate}
\localitemlabels
\item $[A_1, A_2] = 0$ if $\ocal_1,\ocal_2$ are two spacelike separated regions,
and $A_i \in \afk(\ocal_i)$;
\item  $\alpha_{g} \afk(\ocal) = \afk(g.\ocal)$
for all $\ocal,g$ .
\end{enumerate}
We call $\afk$ a \emph{net in a positive energy representation} if, in addition, the $\afk(\ocal)$
are $W^*$ algebras acting on a common Hilbert space $\hcal$, and
\begin{enumerate}
\setcounter{enumi}{2}
\localitemlabels
\item there is a strongly continuous unitary representation $g \mapsto U(g)$ of $\gcal$
on $\hcal$ such that $\alpha_{g} = \ad U(g)$;
\item the joint spectrum of the generators of translations $U(x)$ lies
in the closed forward light cone $\bar \vcal_+$;
\item there exists a vector $\Omega \in \hcal$ which is invariant under all $U(g)$ and
  cyclic for $\afk$.
\end{enumerate}
We call $\afk$ a \emph{net in the vacuum sector} if, in addition, 
\begin{enumerate}
\setcounter{enumi}{5}
\localitemlabels
\item \label{uniqVacCond} the vector $\Omega$ is unique (up to scalar factors) as an invariant vector for the translation group.
\end{enumerate}
\end{definition}

Our approach is to start from a local net $\afk$ in the vacuum sector, with
the Poincaré group $\poincare$ as its symmetry group; this net $\afk$ will be
kept fixed in all that follows. Our aim is to describe the short-distance
scaling limit of $\afk$. Following \cite{BucVer:scaling_algebras}, we define
$\ubfk$ to be the set of bounded functions $\uB: \rbb_+ \to \boundedops$,
$\lambda \mapsto \uB_\lambda$. Equipped with pointwise addition,
multiplication, and $\ast$ operation, and with the norm $\|\uB\| =
\sup_\lambda \|\uB_\lambda\|$, the set $\ubfk$ becomes a \cistar{} algebra. Let
$\ugcal$ be the group formed by Poincaré transformations and dilations; we will
write $\ugcal\ni g=(\mu,x,\Lambda)$ with $\mu \in \rbb_+$, $x \in \rbb^{s+1}$,
and $\Lambda$ a Lorentz matrix. $\ugcal$ acts on $\ubfk$ via a representation
$\ualpha$, given by
\begin{equation}
    (\ualpha_g \uB)_\lambda = \alpha_{\lambda \mu x, \Lambda}(\uB_{\lambda\mu})
    \quad
    \text{for } g=(\mu,x,\Lambda)\in\ugcal, \; \uB \in \ubfk,
\end{equation}
where $\alpha$ is the Poincaré group representation on $\afk$. Note the
rescaling of translations with the scale parameter $\lambda$. We now define new
local algebras as subsets of $\ubfk$:
\begin{equation} 
  \uafk(\ocal) := \big\{ \uA \in \ubfk \,|\, \uA_\lambda \in \afk(\lambda \ocal) \text{ for all }\lambda>0; \; 
   g \mapsto \ualpha_g( \uA ) \text{ is norm continuous} \big\}.
\end{equation}
This is a net of local algebras in the sense of Def.~\ref{def:localnet}, with
the enlarged symmetry group $\ugcal$ \cite{BDM:scaling_fields}. We denote by $\uafk$ the
associated quasilocal algebra, i.e.~the inductive limit of $\uafk(\ocal)$ as $\ocal
\nearrow \rbb^{s+1}$. This $\uafk$ is called the \emph{scaling algebra}. Note
that $\uafk$ has a large center $\zfk(\uafk)$, consisting of all operators $\uA$
of the form $\uA_\lambda = f(\lambda) \idop$, where $f : \rbb_+ \to \cbb$ is a
uniformly continuous function on $\rbb_+$ as a group under multiplication. We
often identify $\uA \in \zfk(\uafk)$ with the function $f$ without further
notice.

For a description of the scaling \emph{limit,} we first consider states on
$\zfk(\uafk)$. Let $\mean$ be a mean on the uniformly continuous
functions\footnote{%
In contrast to \cite{BDM:scaling_fields}, we do not consider means on the
\emph{bounded} functions on $\rbb_+$, but rather on the uniformly continuous
functions. While all of them can be extended to the bounded functions,
these extensions do not play a role in our current investigation.} 
on $\rbb_+$, i.e., a positive normalized linear functional on the
commutative \cistar{} algebra $\zfk(\uafk)$. We say that $\mean$ is
\emph{asymptotic} if $\mean(f) = \lim_{\lambda \to 0} f(\lambda)$ whenever the
limit on the right-hand side exists; or, equivalently, if $\mean(f)=0$ whenever
$f(\lambda)=0$ for small $\lambda$. Asymptotic means are, in this sense,
generalizations of the limit $\lambda \to 0$. Further we consider two important
classes of means:
\begin{enumerate}
  \localitemlabels
  \item \label{it:multMean} $\mean$ is called \emph{multiplicative} if
  $\mean(f g) = \mean(f)\mean(g)$ for all functions $f,g$.
  \item \label{it:invMean} $\mean$ is called \emph{invariant} if
  $\mean(f_\mu)=\mean(f)$ for all functions $f$ and all $\mu>0$, where
  $f_\mu=f(\mu \cdotarg)$.
\end{enumerate}
It is an important fact that \ref{it:multMean} and \ref{it:invMean} are
mutually exclusive; there are no multiplicative invariant means in our
situation (cf.~\cite{Mit:mult_invar_means}).

We now extend these ``generalized limits'' of functions to a limit of operator
sequences, using a projection technique. Let $\omega=
\mtxe{\Omega}{\cdotarg}{\Omega}$ be the vacuum state of $\afk$. This state
induces a projector (or conditional expectation) in $\uafk$ onto $\zfk(\uafk)$,
which we denote by the same symbol:
\begin{equation}
  \omega: \uafk \to \zfk(\uafk), \quad
  (\omega(\uA))_\lambda = \omega(\uA_\lambda) \idop.
\end{equation}
Using this projector, any mean $\mean$ defines a state $\uomega_\mean$ on
$\uafk$ by $\uomega_\mean := \mean \circ \omega$. If here $\mean$ is asymptotic,
we call $\uomega_\mean$ a \emph{limit state}, and typically denote it by
$\uomega_0$. These are the states that correspond to scaling limits of the quantum field
theory. Since there is a one-to-one correspondence between asymptotic means and
limit states, we will usually work with the state $\uomega_0$ only, and not
refer to the mean $\mean$ explicitly. A limit state $\uomega_0$ will be called 
\emph{multiplicative}\footnote{%
For clarity, we note that a multiplicative limit state, by this definition, is
not a multiplicative functional on $\uafk$, but is multiplicative only on the
center $\zfk(\uafk)$.
} 
or \emph{invariant} if the corresponding mean has this
property. Multiplicative limit states correspond to those
considered by Buchholz and Verch in \cite{BucVer:scaling_algebras}. Every other
limit state arises from these by convex combinations and \weakstar limits; this
follows directly from the property of states on the commutative algebra
$\zfk(\uafk)$.

Given a limit state $\uomega_0$, we can obtain the limit
theory via a GNS construction: Let $\pi_0$ be the GNS representation of $\uafk$ with respect to
$\uomega_0$, and $\hcal_0$ the representation space, with GNS vector
$\Omega_0$. Denoting by $\gcal_0$ the subgroup of $\ugcal$ under which
$\uomega_0$ is invariant, we canonically obtain a strongly continuous unitary
representation of $\gcal_0$ on $\hcal_0$ by setting $U_0(g) \pi_0(\uA)\Omega_0
:= \pi_0(\ualpha_g(\uA))\Omega_0$, $g \in \gcal_0$. The subgroup $\gcal_0$
contains the Poincaré group; and if $\uomega_0$ is invariant, then
$\gcal_0=\ugcal$. The translation part of $U_0$ fulfills the spectrum
condition \cite{BDM:scaling_fields}. Setting $\afk_0(\ocal):=
\pi_0(\uafk(\ocal))''$, one obtains a local net $\afk_0$ with symmetry group
$\gcal_0$ in a positive energy representation: the limit theory.

\subsection{Multiplicity of the vacuum state}

If $\uomega_0$ is a multiplicative limit state, its restriction to
$\zfk(\uafk)$ is pure. It has been shown in \cite{BucVer:scaling_algebras}
that in the case $s\geq 2$, this property extends to the entire theory:
$\uomega_0$ is a pure vacuum state on $\uafk$, and $\pi_0$ is an irreducible
representation. On the other hand, if $\uomega_0$ is not multiplicative,
the same must be false, since already $\pi_0\restrict \zfk(\uafk)$ is
reducible. However, we shall show that this property of the center is the only ``source'' of reducibility:
namely one has $\pi_0(\uafk)' = \pi_0(\zfk(\uafk))''$.

We need some preparations to prove this. In the following, set $\zfk_0 :=
\pi_0(\zfk(\uafk))''$, and let $\hcal_\zfk := \clos(\zfk_0 \Omega_0) \subset
\hcal_0$ be the representation space of the commutative algebra.

\begin{lemma} \label{centerProjLemm}
  Let $P_\zfk \in \bcal(\hcal_0)$ be the orthogonal projector onto $\hcal_\zfk$.
  If $s \geq 2$, then $P_\zfk \in \pi_0(\uafk)''$, and $\hcal_\zfk$ is the space of all translation-invariant vectors in $\hcal_0$.
\end{lemma}

\begin{proof}
As a consequence of the spectrum condition in the theory $\afk_0$, it is known
\cite{Ara:observable_alg} that the translation operators $U_0(x)$ are contained in $\pi_0(\uafk)''$. Now let $U_\infty$ be an ultraweak cluster point of $U_0(x)$ as $x$ goes to spacelike infinity on 
some fixed sequence within the time-0 plane. (Such cluster points exist by the Alaoglu-Bourbaki theorem.) 
Then $U_\infty \in \pi_0(\uafk)''$; we will show $U_\infty = P_\zfk$. 

To that end, we first note that
\begin{equation} \label{meanfactor}
   \uomega_0(\uA \,\uB) = \uomega_0( \omega(\uA) \, \uB ) \quad
 \text{for all } \uA \in \uafk, \; \uB \in \zfk(\uafk),
\end{equation}
which follows directly from the definition of $\uomega_0$. 
Now we make use of the cluster property of the vacuum at finite scales. 
As in \cite[Lemma~4.3]{BucVer:scaling_algebras},
one can obtain the following norm estimate in the algebra $\uafk$:
\begin{equation}
   \|\omega(\uA \, \ualpha_x \uB)  - \omega(\uA) \omega(\uB) \|
  \leq
   c \frac{r^{s}}{|x|^{s-1}} \big( \|\uA\| \|\dot \uB \| + \|\dot \uA\| \| \uB
   \|\big)
\end{equation}
for fixed $r > 0$, $x$ in the time-0 plane with $|x| > 3r$, and for $\uA,\uB$
chosen from some norm-dense subset of $\uafk(\ocal_r)$, with $\ocal_r$ being the standard double cone of
radius $r$ around the origin. Here $c>0$ is some constant, and the dot denotes
the time derivative. This implies that as $|x| \to \infty$,
\begin{equation}
  \lim_x \uomega_0( \uA \,\ualpha_x \uB ) = \uomega_0 (\omega(\uA) \omega(\uB))
\end{equation}
for these $\uA, \uB$. Now it follows from Eq.~\eqref{meanfactor} -- with $\omega(\uB)$ in place of $\uB$ -- 
that
\begin{equation}
  \hrskp{ \pi_0(\uA) \Omega_0 }{ U_\infty \pi_0( \uB) \Omega_0 }
  = \lim_x \uomega_0( \uA \, \ualpha_x \uB ) = \hrskp{ \pi_0(\uA) \Omega_0 }{
  \pi_0( \omega(\uB)) \Omega_0 }.
\end{equation}
Continuing this relation from the dense sets chosen, this means
\begin{equation}
  U_\infty \pi_0( \uB) \Omega_0 = \pi_0(\omega(\uB)) \Omega_0 
 \quad \text{for all } \uB \in \uafk.
\end{equation}
This shows that $U_\infty^2 = U_\infty$, and $\img U_\infty = \hcal_\zfk$. Also, again applying
Eq.~\eqref{meanfactor}, one obtains $U_\infty\st = U_\infty$. Thus $U_\infty$ is the unique orthogonal 
projector onto $\hcal_\zfk$.---For the last part, note that translations act trivially on $\hcal_\zfk$,
and that $U_\infty$ leaves all translation-invariant vectors unchanged; so $\hcal_\zfk$ is the space
of all translation-invariant vectors. 
\end{proof}

We are now ready to prove the announced result about the commutant of $\pi_0(\uafk)$.

\begin{theorem} \label{thm:pi0CommThm}
Let $s \geq 2$. Let $\uomega_0$ be a limit state, and let $\pi_0$ be the
corresponding GNS representation. Then $\pi_0(\uafk)' = \pi_0(\zfk(\uafk))''$.
\end{theorem}
\begin{proof}
Let $B \in \pi_0(\uafk)'$. By Lemma~\ref{centerProjLemm}, $B$ commutes with $P_\zfk$;
hence $B \hcal_\zfk \subset \hcal_\zfk$, and $B \restrict \hcal_\zfk \in \bfk(\hcal_\zfk)$ is well-defined.
As $\zfk_0 \subset \pi_0(\uafk)''$, we know that
\begin{equation}
  [ B \restrict \hcal_\zfk , C \restrict \hcal_\zfk] = 0 \quad
  \text{for all } C \in \zfk_0
\end{equation}
as an equation in $\bfk(\hcal_\zfk)$. Since $\zfk_0\restrict \hcal_\zfk$ is a maximal abelian algebra in
$\bfk(\hcal_\zfk)$ \cite[Lemma~4.3.15]{BraRob:qsm1}, there exists $C \in \zfk_0$ with $B \restrict \hcal_\zfk = C \restrict \hcal_\zfk$.
Now for any $A \in \pi_0(\uafk)''$, we can compute
\begin{equation}
   BA \Omega_0  =  AB \Omega_0 = AC \Omega_0 = CA \Omega_0.
\end{equation}
Since $\Omega_0$ is cyclic for $\pi_0$, this implies $B = C$. Thus
$\pi_0(\uafk)' \subset \zfk_0$. The reverse inclusion is trivial.
\end{proof}

It should be noted that the same theorem does not hold in 1+1 space-time dimensions. In this case, it is known
even in free field theory \cite[Sec.~4]{BucVer:scaling_examples} 
that the algebra $\pi_0(\uafk)$ has a large center, even if $\pi_0(\zfk(\uafk))
= \cbb\idop$.

We can now easily reproduce the known results for multiplicative limit
states. In this case, the GNS representation of the abelian algebra
$\zfk(\uafk)$ for the state $\uomega_0$ must be irreducible; thus $\zfk_0 = \cbb
\idop$, and $\dim \hcal_\zfk=1$. The above results imply:

\begin{corollary}
  Let $\uomega_0$ be a multiplicative limit state, and let $s \geq 2$. Then
  $\Omega_0$ is unique up to a scalar factor as an invariant vector for the
  translations $U_0(x)$, and the representation $\pi_0$ is irreducible.
  $\afk_0$ is a net in the vacuum sector in the sense of
  Def.~\ref{def:localnet}.
\end{corollary}

\subsection{Wedge algebras and geometric modular action}

While we have defined the scaling limit in terms of local algebras for bounded
regions, it is also worthwhile to consider algebras associated with unbounded,
in particular wedge-shaped regions. This is particularly important in the
context of charge analysis for the limit theory
\cite{DMV:scaling_charges,DAnMor:supersel_models}.
While we do not enter this topic here, and do not build on it in the following,
we wish to discuss briefly how wedge algebras and the condition of geometric modular
action fit into our context. Again, this transfers results of
\cite{BucVer:scaling_algebras} to our generalized class of limit states.

Let $\wcal$ be a wedge region, i.e.\  $\wcal$ is a Poincar\'e transform of the right wedge 
\begin{equation}
\wcal_+ = -\wcal_- = \{ x \in \rbb^4\,|\, x \cdot e_+ < 0 \}, \qquad
\text{where } e_\pm := (\pm 1, 1,0,0).
\end{equation}
Note that $(\overline{\wcal}_+)' = \wcal_-$. We introduce the one-parameter
group $(\Lambda_t)_{t \in \rbb}$ of Lorentz boosts leaving $\wcal_+$ invariant,
fixed by $\Lambda_t e_\pm = \exp(\pm t) e_\pm$, and acting as the identity on
the edge $(e_\pm)^\perp$ of $\wcal_+$. Let furthermore $j$ be the inversion with respect to the edge of $\wcal_+$, i.e.\ $j e_\pm = - e_\pm$ and $j = \idop$ on $(e_\pm)^\perp$. Note that $j^2 = \idop$ and $j \wcal_+ = \wcal_-$.

For a local net of algebras (resp.\ for a net of algebras in a positive energy
representation) $\ocal \mapsto \afk(\ocal)$, we define the algebra
$\afk(\wcal)$ associated to the wedge $\wcal$ as the \cistar-algebra (resp.
\wstar-algebra) generated by the algebras $\afk(\ocal)$, where $\ocal$ is any
double cone whose closure is contained in $\wcal$ ($\ocal \subset\subset \wcal$
in symbols). With these definitions, we can adapt the arguments
in~\cite[Lemma 6.1]{BucVer:scaling_algebras}, which do not depend on
irreducibility of the net. It is then straightforward to verify that, for a net
$\afk$ in a positive energy representation, the vacuum vector $\Omega$ is cyclic and separating for all wedge algebras $\afk(\wcal)$. This allows us to introduce the notion of geometric modular action.

\begin{definition} \label{def:modularGeoAction}
Let $\afk$ be a local net in a positive energy
representation, and denote by $\Delta$, $J$ the modular objects associated to
$\afk(\wcal_+)$, $\Omega$.
The net $\afk$ is said to satisfy the \emph{condition of geometric
modular action} if there holds
\begin{alignat*}{2}
\Delta^{it} &= U(\Lambda_{2\pi t}),& \qquad \qquad t &\in \rbb,
\\
J U(x,\Lambda) J &= U(jx,j \Lambda j),& (x,\Lambda) &\in
\pcal_+^\uparrow,
\\ \tag{$\ast$} 
J \afk(\ocal) J &= \afk(j \ocal). \label{cgmaInversion}
\end{alignat*}
\end{definition}

If $\afk$ satisfies the condition of geometry modular action, then it also satisfies wedge duality, since, according to Tomita-Takesaki theory and equation~\eqref{cgmaInversion},
\begin{equation}
\afk(\wcal_+)' = J\afk(\wcal_+) J = \afk(\wcal_-).
\end{equation}

This also implies that $\afk$ satisfies \emph{essential Haag duality}, i.e.\ that the \emph{dual net} $\afk^d$ of $\afk$, defined on double cones $\ocal$ as
\begin{equation}
\afk^d(\ocal) := \bigwedge_{\wcal \supset \ocal}\afk(\wcal),
\end{equation}
is local and such that $\afk(\ocal) \subset \afk^d(\ocal)$ for each double cone $\ocal$.

From now on, let $\afk$ be a net in the vacuum sector, and $\uomega_0$ a scaling limit state, with $\pi_0$ the corresponding scaling limit representation. It holds that $\pi_0(\uafk(\wcal))'' = \afk_0(\wcal)$, since clearly
\begin{equation}
\pi_0(\uafk(\wcal)) = \overline{\bigcup_{\ocal\subset\subset\wcal}\pi_0(\uafk(\ocal))},
\end{equation}
and therefore
\begin{equation}
\pi_0(\uafk(\wcal))' = \bigwedge_{\ocal \subset\subset\wcal} \pi_0(\uafk(\ocal))' = \bigwedge_{\ocal \subset\subset\wcal} \afk_0(\ocal)' = \afk_0(\wcal)'.
\end{equation} 

\begin{proposition}
Assume that $\afk$ satisfies the condition of geometric modular action. Then for each limit state $\uomega_0$, the corresponding limit theory $\afk_0$ also satisfies the condition of geometric modular action.
\end{proposition}

\begin{proof}
It's a straightforward adaptation of the proofs of Lemma~6.2 and
Proposition~6.3 of~\cite{BucVer:scaling_algebras}. The only point which is worth mentioning is the proof that $\uomega_0$ is a KMS state (at inverse temperature $2\pi$)
for the algebra $\uafk(\wcal_+)$ with respect to the one-parameter group of
automorphisms $(\ualpha_{\Lambda_t})_{t\in \rbb}$, which goes as follows. Let
$\mean$ be the mean which induces $\uomega_0$. Then $\mean$ is a
\weakstar limit of convex combinations of multiplicative means, 
and therefore $\uomega_0$ is a \weakstar limit of convex combinations 
of multiplicative limit states. For such states, the arguments
in~\cite[Lemma~6.2]{BucVer:scaling_algebras} show that they are KMS on $\uafk(\wcal_+)$, and therefore, the set of KMS states at a fixed inverse temperature being convex
and \weakstar closed~\cite[Thm.\ 5.3.30]{BraRob:qsm2}, this holds also for
$\uomega_0$.
\end{proof}

\section{Decomposition theory} \label{sec:decompTheory}

Our aim is now to decompose an arbitrary limit state $\uomega_0$ into ``simple''
limit states of the Buchholz-Verch type, and to obtain corresponding
decompositions of the relevant objects in the limit theory. We start by proving an integral
decomposition which is a consequence of standard results.

\begin{proposition} \label{pro:stateDecomp}
Let $\uomega_0$ be a limit state. There exists a compact Hausdorff space
$\zcal$, a regular Borel probability measure $\nu$ on $\zcal$,
and for each $z \in \zcal$ a multiplicative limit state $\uomega_z$, such that
\begin{equation*}
  \uomega_0 (\uA) = \int_\zcal d\nu(z) \,\uomega_z(\uA)
  \quad \text{for all } \uA \in \uafk.
\end{equation*}
Further, the map $\zfk(\uafk) \to \ccal(\zcal)$, $\uC \mapsto (z \mapsto
\uomega_z(\uC))$ is surjective.
\end{proposition}

\begin{proof}
Let $\pi_0$ be the GNS representation of $\uafk$ for $\uomega_0$. Consider the
\cistar{} algebra $\pi_0(\zfk(\uafk))$. It is well known that this commutative
algebra is isomorphic to $\ccal(\zcal)$ for a compact Hausdorff space $\zcal$, with the
isomorphism being given by $\pi_0(\uC) \mapsto (z \mapsto \rho_z(\pi_0(\uC)))$,
where the $\rho_z$ are multiplicative functionals. Now by the Riesz
representation theorem, the GNS state $\mtxe{\Omega_0}{\cdotarg}{\Omega_0}$
on $\pi_0(\zfk(\uafk)) \isom \ccal(\zcal)$ is given by a regular Borel measure
$\nu$ on $\zcal$. Explicitly, one has for all $\uC \in \zfk(\uafk)$, 
\begin{equation} \label{eq:centerDecomp}
 \uomega_0(\uC) = \mtxe{\Omega_0}{\pi_0(\uC)}{\Omega_0}
 = \int_\zcal d\nu(z) \, \rho_z \circ \pi_0(\uC).
\end{equation}
It is clear that $\nu(\zcal)=1$. In the above expression, $\mean_z := \rho_z
\circ \pi_0$ are multiplicative means; they are asymptotic, since $\pi_0(\uA)=0$ whenever $\uA_\lambda$
vanishes for small $\lambda$. Thus, setting $\uomega_z=\rho_z \circ \pi_0 \circ
\omega$ as usual, we obtain multiplicative limit states $\uomega_z$ on $\uafk$
such that
\begin{equation}
 \uomega_0(\uA) = \int_\zcal d\nu(z) \, \uomega_z(\uA)
 \quad \text{for all }\uA \in \uafk.
\end{equation}
As a last point, the map $\zfk(\uafk) \to \ccal(\zcal)$, $\uC \mapsto
(z \mapsto \uomega_z(\uC)) = (z \mapsto \rho_z(\pi_0(\uC)))$ is
surjective by construction.
\end{proof}

We have thus decomposed a general limit state $\uomega_0$ into
multiplicative limit states $\uomega_z$. In the case $s \geq 2$, this will also
be a decomposition into pure states; but the above result does not depend on that.
Also, we emphasize that our aim is not a decomposition of the von Neumann
algebra $\pi_0(\uafk)''$ along its center; rather we work on the \cistar{} algebraic side
only.

We would now like to interpret the above decomposition in the sense of
decomposing the limit Hilbert space $\hcal_0$ as a direct integral. This is
complicated by the fact that our measure spaces $(\zcal,\nu)$ can be of a very
general nature, making the limit Hilbert space nonseparable. In fact, if
$\uomega_0$ is an \emph{invariant} limit state, 
one finds that all vectors of the form $\pi_0(\uC) \Omega_0$ are mutually
orthogonal if $\uC_\lambda=\chi(\lambda)\idop$, where $\chi$ is a
character on $\rbb_+$. Since there are clearly uncountably many
characters -- just take $\chi(\lambda) = \lambda^{\imath k}$ with $k
\in \rbb$ -- the limit Hilbert space $\hcal_0$ cannot be separable in this case.

The theory of direct integrals of Hilbert spaces in the absence of separability
assumptions is nonstandard and only partially complete; we give a brief review
in Appendix~\ref{sec:diApp}. Here we note that the notion of a direct
integral over $\zcal$, with fiber spaces $\hcal_z$, crucially depends on the
specification of a \emph{fundamental family} $\Gamma \subset
\prod_{z\in\zcal} \hcal_z$. This $\Gamma$ is a vector space with certain extra
conditions (see Def.~\ref{def:fundamentalFamily}) that serves to define which Hilbert space
valued functions are considered measurable. Indeed, using
the exact notions, we prove:

\begin{theorem} \label{thm:hilbertDecomp}
 Let $\uomega_0$ be a limit state, and $\zcal$, $\nu$, $\uomega_z$ as in
 Proposition~\ref{pro:stateDecomp}. Let $\pi_z$, $\hcal_z$, $\Omega_z$ be
 the GNS representation objects corresponding to $\uomega_z$. Then,
 \begin{equation*}
    \Gamma := \{ z \mapsto \pi_z(\uA) \Omega_z \,|\, \uA \in \uafk\} \subset
    \prod_{z\in\zcal}   \hcal_z
 \end{equation*}
is a fundamental family. With respect to this family, it holds that
\begin{equation*}
  \hcal_0 \isom \int_\zcal^\Gamma d\nu(z) \,\hcal_z,
\end{equation*}
where the isomorphism is given by
\begin{equation*}
  \pi_0(\uA) \Omega_0 \mapsto \int_\zcal^\Gamma d\nu(z)\, \pi_z(\uA)
  \Omega_z, \quad
  \uA \in \uafk.
 \end{equation*}
\end{theorem}

\begin{proof}
It is clear that $\Gamma$ is a linear space; and per
Prop.~\ref{pro:stateDecomp}, the function $z \mapsto
\|\pi_z(\uA)\Omega_z\|^2= \uomega_z(\uA\st\uA)$ is integrable for any $\uA \in
\uafk$. Thus $\Gamma$ is a fundamental family per Definition~\ref{def:fundamentalFamily}.
The map
\begin{equation} \label{eq:DecompIso}
   W: \hcal_0 \to \prod_{z \in \zcal} \hcal_z, \quad \pi_0(\uA) \Omega_0 \mapsto
   (z \mapsto \pi_z(\uA) \Omega_z)
\end{equation} 
is clearly linear and isometric when $\uA$ ranges through $\uafk$; thus $W$ can
in fact be extended to a well-defined isometric map from $\hcal_0$ into
$\bar \Gamma$. It remains to show that $W$ is surjective. In
fact, since $L^\infty(\zcal) \cdot \Gamma$ is total in the direct integral space, it
suffices to show that all vectors of the form
\begin{equation}
   \int_\zcal^\Gamma d\nu(z) f(z)\, \pi_z(\uA) \Omega_z, \quad
   f \in L^\infty(\zcal), \; \uA \in \uafk,
\end{equation}
can be approximated in norm with vectors of the form $W \pi_0(\uB) \Omega_0$,
$\uB \in \uafk$.

To that end, let $f \in L^\infty(\zcal)$ and $\uA \in \uafk$ be fixed. We first
note that, as a simple consequence of Lusin's theorem, there exist functions $f_n \in \ccal(\zcal)$ 
such that $\|f_n\|_\infty \leq \|f\|_\infty$ and $\lim_{n \to \infty} f_n(z) = f(z)$ 
for almost every $z \in \zcal$. On the other hand, per
Proposition~\ref{pro:stateDecomp} there exist $\uC_n \in \zfk(\uafk)$ such that
$\uomega_z(\uC_n) = f_n(z)$ for all $z \in \zcal$, which implies $\pi_z(\uC_n)
= f_n(z) \idop$. Therefore we have
\begin{equation}
\lim_{n \to \infty} \int_\zcal \| 
(f(z)\pi_z(\uA) - \pi_z (\uC_n \uA))\Omega_z \|^2 d\nu(z) = 0
\end{equation}
by an application of the dominated convergence theorem.
\end{proof}

In the following, we will usually not denote the above isomorphism explicitly,
but rather identify $\hcal_0$ with its direct integral representation.
In this way, the subspace $\hcal_\zfk \subset \hcal_0$ is isomorphic to
the function space $L^2(\zcal,\nu)$, where $f \in L^2(\zcal,\nu)$ 
is identified with $\int^\Gamma_\zcal d\nu(z) f(z) \Omega_z \in \hcal_0$. 
The next corollary follows directly from the proof above, since a decomposition
of operators needs to be checked on the fundamental family only (Lemma~\ref{lem:measOnGamma}).

\begin{corollary} \label{cor:reprDecomp}
  With respect to the direct integral decomposition in
  Theorem~\ref{thm:hilbertDecomp}, all operators $\pi_0(\uA)$, $\uA \in \uafk$
  are decomposable, and one has $\pi_0 = \int_\zcal^\Gamma d\nu(z) \pi_z$.
  If $\uA \in \zfk(\uafk)$, then $\pi_0(\uA)$ is diagonal, with $\pi_0(\uA) =
  \int_\zcal^\Gamma d\nu(z)\, \uomega_z(\uA) \idop$.
\end{corollary}

Finally, we remark that Lorentz symmetries $U_0(x,\Lambda)$ in the limit theory
are decomposable.

\begin{proposition} \label{pro:poincareDecomp}
  Let $g \mapsto U_z(g)$ be the implementation of $\poincare$ on the limit
  Hilbert space $\hcal_z$ corresponding to $\uomega_z$. Then, one has
  \begin{equation*}
    U_0(g) = \int_\zcal^\Gamma d\nu(z) U_z(g) \quad
    \text{for all } g \in \poincare.
  \end{equation*}
\end{proposition}

\begin{proof}
Again, it suffices to verify this on vectors from $\Gamma$. With $W$ being the
isomorphism introduced in the proof of Theorem~\ref{thm:hilbertDecomp}, one
obtains for all $g \in \poincare$ and $\uA\in\uafk$,
\begin{equation}
  W U_0(g) \pi_0(\uA) \Omega_0 = W \pi_0(\ualpha_g \uA) \Omega_0
  = \int_\zcal^\Gamma d\nu(z) \pi_z(\ualpha_g \uA) \Omega_z
  = \int_\zcal^\Gamma d\nu(z) U_z(g) \pi_z(\uA) \Omega_z.
\end{equation}
This proves the proposition.
\end{proof}

It should be remarked that the same simple structure cannot be expected for
\emph{dilations}, if they exist as a symmetry of the limit. For even if
$\uomega_0 \circ \ualpha_\mu = \uomega_0$, the multiplicative limit states
$\uomega_z$ cannot be invariant under $\ualpha_\mu$, not even when
restricted to $\zfk(\uafk)$. Thus, the unitaries $U_0(\mu)$ will not commute
with $\pi_0(\zfk(\uafk))$, and can therefore not be decomposable. In special
situations, there may be a generalized sense in which the dilation unitaries
can be decomposed; we will investigate this in more detail in
Sec.~\ref{sec:dilations}.

\section{Unique and factorizing scaling limits} \label{sec:uniqueFactorLim}

The limit theory on the Hilbert space $\hcal_0$ is composed, as
discussed in the previous section, of simpler components that live on the
``fibre'' Hilbert spaces $\hcal_z$ of the direct integral. It is natural to ask
whether the theories on these spaces $\hcal_z$, or more precisely, the nets of algebras $\afk_z(\ocal) =
\pi_z(\uafk(\ocal))''$, are similar or identical in a certain sense. While no
models have been explicitly constructed for which the limit theories
substantially depend on the choice of a (multiplicative) limit
state,\footnote{See however
\cite[Sec.~5]{Buc:phase_space_scaling} for some ideas to that end.}
it does not seem to be excluded that measurable properties, such as the
mass spectrum or charge structure of $\afk_z$, can depend on $z$. 

For most
applications in physics, however, one expects that the situation is simpler,
and that the limit theory does not depend substantially on the choice of
$\uomega_z$. Here it would be much too strict to require that the
representations $\pi_z$ are unitarily equivalent. [In fact, for
$s \geq 2$, the $\pi_z$ are irreducible per Thm.~\ref{thm:pi0CommThm}, and since
they do not agree on $\zfk(\uafk)$, they are even pairwise disjoint.]
Rather one can expect that their images, the algebras $\afk_z(\ocal)$, are unique \emph{as sets,} 
up to unitaries that identify the different Hilbert spaces $\hcal_z$; 
see Def.~\ref{def:uniqueLimit} below. This is the situation of a \emph{unique
scaling limit} in the sense of Buchholz and Verch.
 
In the present section, we want to elaborate how the situation of unique
scaling limits, originally formulated for multiplicative limit states, fits
into our generalized context. To that end, we will formulate several
conditions on the limit theory that roughly correspond to unique limits, and
discuss their mutual dependencies. 

\subsection{Definitions}

We shall first motivate and define the conditions to be considered; the proofs
of their interrelations are deferred to sections further below. We start by
recalling the condition of a unique scaling limit in the sense of
\cite{BucVer:scaling_algebras}, with some slight modifications.

\begin{definition} \label{def:uniqueLimit}
The theory $\afk$ is said to have a \emph{unique scaling limit} if there exists
a local Poincaré covariant net $(\qafk, \qhcal, \qOmega, \qU)$ in
the vacuum sector such that the following holds. For every \emph{multiplicative}
limit state $\uomega_0$, there exists a unitary $V: \hcal_0 \to \qhcal$ such that
$V \Omega_0 = \qOmega$, $V U_0(g) V\st = \qU(g)$ for all $g \in
\poincare$, and $V \afk_0(\ocal) V\st \subset \qafk(\ocal)$ for all open
bounded regions $\ocal$. 
\end{definition}

This includes the
aspect of a ``unique vacuum structure''. Compared with
\cite{BucVer:scaling_algebras}, we have somewhat weakened the condition,
since we require only inclusion of $V \afk_0(\ocal) V\st$ in $\qafk(\ocal)$, not equality. This is for the following reason.
Supposing that both $\afk$ and $\qafk$ fulfill the condition of geometric
modular action (Definition~\ref{def:modularGeoAction}), such that the Haag-dualized nets of
$\afk_0$ and $\qafk$ are well-defined, our condition precisely implies that
these dualized nets agree for all multiplicative limit states. Since for many
applications, particularly charge analysis \cite{DMV:scaling_charges}, the dualized limit nets are seen as the fundamental
objects, we think that this is a reasonable generalization of the condition.

For a general, not necessarily multiplicative limit state $\uomega_0$, we
obtain a decomposition $\uomega_0 = \int_\zcal d\nu(z) \uomega_z$ into
multiplicative states, as discussed in Sec.~\ref{sec:decompTheory}, and thus obtain from
Def.~\ref{def:uniqueLimit} corresponding unitaries $V_z$ for every $z$. 
Due to the very general nature of the measure
space $\zcal$, and due to a possible arbitrariness in the choice of $V_z$,
particularly if $\qafk$ possesses inner symmetries, an analysis of
$\uomega_0$ seems impossible in this generality. Rather we will often make use
of a regularity condition, which is formulated as follows.

\begin{definition} \label{def:stateRegular}
Suppose that the theory $\afk$ has a unique scaling limit. We say that a limit
state $\uomega_0$ is \emph{regular} if there is a choice of the unitaries
$V_z$ such that for any $\uA \in \uafk$, the function
$\varphi_{\uA}: \zcal \to \qhcal$, $z \mapsto V_z \pi_z(\uA) \Omega_z$
is Lusin measurable [i.e., is contained in $L^2(\zcal,\nu,\qhcal)$].
\end{definition}

We shall later give a sufficient condition for the above regularity, 
which actually implies that the functions $\varphi_{\uA}$ in fact
be chosen \emph{constant} in generic cases.

Our concepts so far refer to multiplicative limit states mostly. We will
now give a generalization of Def.~\ref{def:uniqueLimit} that involves
generalized limit states directly, and that seems natural in our context. It is
based on the picture that the limit Hilbert space should have a tensor product
structure, $\hcal_0 \isom \hcal_\zfk \otimes \qhcal$, where $\qhcal$ is the
unique representation space associated with multiplicative limit states, and
$\hcal_\zfk$ is the representation space of $\zfk(\uafk)$ under $\pi_0$. All
objects of the theory -- local algebras, Poincaré symmetries, and the vacuum
vector -- should factorize along this tensor product. We now formulate this in
detail.

\begin{definition} \label{def:factorLimit}
The theory $\afk$ is said to have a \emph{factorizing scaling limit} if
there exists a local Poincaré covariant net $(\qafk, \qhcal, \qOmega, \qU)$ in
the vacuum sector such that the following holds. For every limit
state $\uomega_0$, there exists a decomposable unitary $V: \hcal_0 \to
L^2(\zcal,\nu,\qhcal)$, $V = \int_\zcal^{\Gamma,\oplus} d\nu(z) V_z$ with
unitaries $V_z: \hcal_z \to \qhcal$, such that $V \Omega_0 = \Omega_\zfk \otimes \qOmega$, $V U_0(g)
V\st = \idop \otimes \qU(g)$ for all $g \in \poincare$, and
$V_z \afk_z(\ocal) V_z\st \subset \qafk(\ocal)$ for all open bounded regions $\ocal$ and all $z \in \zcal$.
\end{definition}

Here $\Omega_\zfk \in \hcal_\zfk$ denotes the GNS vector of the commutative
algebra. The conditions on local algebras are deliberately chosen quite strict.
We require $V_z \afk_z(\ocal) V_z\st \subset \qafk(\ocal)$ for every $z$, rather than the weaker condition $V \afk_0(\ocal) V\st \subset \zfk_0 \bar\otimes \qafk(\ocal)$. This serves to avoid countability
problems; see Sec.~\ref{sec:factorToUnique} for further discussion.

In subsequent sections, we will show that the notion of a unique scaling limit
and a factorizing scaling limit are \emph{cum grano salis} identical, up to the extra regularity condition in Definition~\ref{def:stateRegular} that we
have to assume.

We also consider a stronger condition, which is easier to check in models.
Our ansatz is to require a sufficiently large subset
$\uafk_\mathrm{conv} \subset\uafk$ such that for each $\uA \in
\uafk_\mathrm{conv}$, the function $\lambda \mapsto
\omega(\uA_\lambda)$ is convergent as $\lambda \to 0$. Consider the
following definition:

\begin{definition} \label{def:convLimit}
The theory $\afk$ is said to have a \emph{convergent scaling limit} if
there exists an $\ualpha$-invariant \cistar{} subalgebra $\uafk_\mathrm{conv}
\subset \uafk$ with the following properties:
\begin{enumerate}
\localitemlabels
\item 
For each $\uA \in \uafk_\mathrm{conv}$, the function $\lambda \mapsto
\omega(\uA_\lambda)$ converges as $\lambda \to 0$.
\item\label{it:convWeaklyDense} 
If $\uomega_0$ is a \emph{multiplicative} limit state, then
$\pi_0(\uafk(\ocal)\cap\uafk_\mathrm{conv})$ is weakly dense in
$\afk_0(\ocal)$ for every open bounded region $\ocal$.
\end{enumerate}
\end{definition}

It follows directly from \ref{it:convWeaklyDense} that also
$\pi_0(\uafk_\mathrm{conv})\Omega_0$ is dense in $\hcal_0$.
The condition roughly says that ``convergent scaling functions'' are
sufficient for describing the limit theory -- considering nonconvergent
sequences is only required for technical consistency of our formalism, for
describing the image of $\zfk(\uafk)$, which does not directly relate to 
quantum theory. This is heuristically expected in many physical models: In
usual renormalization approaches in formal perturbation theory, the selection of
subsequences or filters to enforce convergence seems not to be
widespread, and sequences of pointlike fields can be chosen to converge in
matrix elements.

We will show that the above condition is sufficient for the scaling limit to be
unique, and all limit states to be regular. In fact, we shall see later
that also the structure of dilations simplifies.

\begin{figure} 
\caption{Implications between the conditions on the limit theory.
Arrows marked with $\ast$ are only proven under additional separability assumptions.}
\label{fig:conditions}
\begin{center}
\mbox{
\xymatrix{ 
&
  *+[F]\txt{Unique limit}
\\
  *+[F]\txt{Convergent limit}
  \ar@{=>}@(u,r)[ur]
  \ar@{=>}@(d,r)[dr]^(.3){\ast}
  &
    *+[o][F-]{+}
  \ar@{--}[d]
  \ar@{--}[u]
&
  *+[F]\txt{Factorizing limit}
  \ar@{<=}[l]_(.7)\ast
  \ar@{=>}@(d,r)[dl]_(.3){\ast}
  \ar@{=>}@(u,r)[ul]
\\
 &
    *+[F]\txt{Regularity condition}
}}
\end{center}
\end{figure}
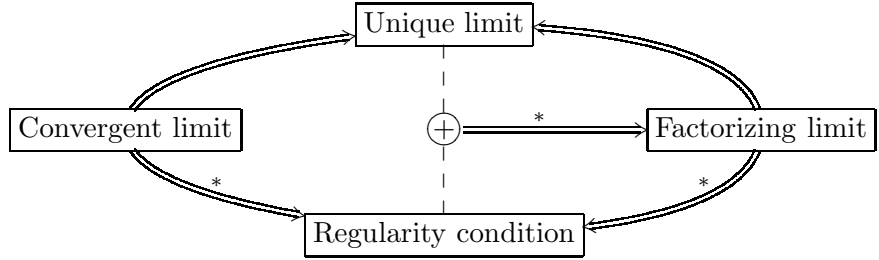

Figure~\ref{fig:conditions} summarizes the different
conditions we introduced, and shows the implications we briefly mentioned. We
will now go ahead and prove that the individual arrows are indeed correct.
However, in order to avoid problems with the direct integral spaces involved,
we shall make certain separability assumptions in most cases. Let us comment on
these. For multiplicative limit states, it seems a reasonable assumption that
the limit Hilbert space $\hcal_0$ is separable. This would follow, from example,
from the Haag-Swieca compactness condition; cf. \cite{Buc:phase_space_scaling}. For general limit states, in
particular if these are invariant, $\hcal_0$ cannot be separable since already
$\hcal_\zfk \isom L^2(\zcal,\nu)$ is nonseparable, as discussed in
Sec.~\ref{sec:decompTheory}. We can however reasonably assume that $\hcal_0$
fulfills a condition which we call \emph{uniform separability}; cf.
Def.~\ref{def:fundamentalSequence} in the appendix. This means that a countable set
$\{\chi_j\} \subset \hcal_0=\int_\zcal^\Gamma d\nu(z)\hcal_z$ exists such
that $\{\chi_j(z)\}$ is dense in every $\hcal_z$. As we shall see in
Sec.~\ref{sec:compactness}, uniform separability follows from a
sharpened version of the Haag-Swieca compactness condition; and we will show in
Sec.~\ref{sec:examples} that this compactness condition is indeed
fulfilled in relevant examples.

\subsection[Unique limit => factorizing limit]{Unique limit
$\Rightarrow$ factorizing limit}

In the following, we suppose that $\afk$ has a unique scaling limit. We
fix a regular limit state $\uomega_0$, and denote the associated objects $\zcal,
\nu, \hcal_0, \pi_0, \Omega_0, \hcal_z, \pi_z, \Omega_z, V_z$ as usual. 
In order to prove that the scaling limit factorizes, we have to construct a
unitary $V:\hcal_0\to L^2(\zcal,\nu,\qhcal)$ with appropriate properties.
In fact, this $V$ is intuitively given by $V = \int_\zcal^{\Gamma,\oplus}
d\nu(z) V_z$; and the key question turns out to be whether this $V$ is surjective. We will prove this only
under separability assumptions.

\begin{proposition} \label{pro:vWelldef}
Let $\afk$ have a unique scaling limit; let $\uomega_0$ be a regular limit
state; and suppose that $\hcal_0$ is uniformly separable. Then,
\begin{equation*}
  V : \hcal_0 \to L^2(\zcal,\nu,\qhcal), \quad
  V = \int_\zcal^{\Gamma,\oplus} d\nu(z) V_z
\end{equation*}
defines a unitary operator.
\end{proposition}

\begin{proof}
First, it is clear that if $\hcal_0$ is uniformly separable, then all
$\hcal_z$, and in particular $\qhcal$, are separable. Hence
$L^2(\zcal,\nu,\qhcal)$ is uniformly separable. 

Now note that $V$ is well-defined precisely by the regularity condition.
Further, writing explicitly
\begin{equation}
  V \pi_0(\uA) \Omega_0 = \big( z \mapsto V_z \pi_z(\uA) \Omega_z \big),
  \quad \uA \in \uafk,
\end{equation}
one has
\begin{equation}
  \| V \pi_0(\uA) \Omega_0 \|^2 
  = \int_\zcal  d\nu(z) \, \|V_z \pi_z(\uA)\Omega_z \|^2
  = \int_\zcal  d\nu(z) \, \| \pi_z(\uA)\Omega_z \|^2 
  = \|\pi_0(\uA) \Omega_0  \|^2,
\end{equation}
so $V$ is isometric. It remains to show that $V$ is surjective.
To that end, let $P$ be the orthogonal projector onto $\img V$. Since $V$
commutes with all diagonal operators, so does $P$; thus $P$ is decomposable:
$P = \int_\zcal^\oplus d\nu(z) P(z)$. Now compute
\begin{equation}
  0 = (1-P) V = \int_\zcal^{\Gamma,\oplus} d\nu(z) (1-P(z)) V_z.
\end{equation}
Using uniform separability of \emph{both} spaces involved, we obtain that
$(1-P(z))V_z=0$ a.~e. Since the $V_z$ are surjective onto $\qhcal$, this means
$P(z)=\idop$ a.~e. This implies $P=\idop$, so $V$ is surjective.
\end{proof}

It is clear that $V \Omega_0 = \Omega_\zfk \otimes \qOmega$; and we can also
verify from the properties of the $V_z$ with respect to Poincaré symmetries that
\begin{equation}
  V U_0(g) V\st = \idop \otimes \qU(g) \quad
  \text{for all } g \in \poincare.
\end{equation}
Also, by the definition of the unique scaling limit, it must hold that $V_z
\afk_z(\ocal) V_z\st \subset \qafk(\ocal)$ for all $z$. Summarizing the results
of this section, we have shown:

\begin{theorem} \label{thm:uniqueImpliesFactor}
Suppose that $\afk$ has a unique scaling limit, that every limit state
$\uomega_0$ is regular, and that the limit spaces $\hcal_0$ are uniformly separable.
Then the scaling limit of $\afk$ is factorizing.
\end{theorem}

\subsection[Factorizing limit => unique limit]{Factorizing limit
$\Rightarrow$ unique limit}
\label{sec:factorToUnique}

Now reversing the arrow, we start from a theory with factorizing scaling limit,
and want to show that the scaling limit is unique in the sense of Buchholz and
Verch, and that the limit states are regular. At first glance, this implication
seems to be apparent from the definitions. A detailed investigation however reveals some subtleties, which
again lead us to making separability assumptions.

\begin{theorem} \label{thm:factorToUnique}
   Assume that $\afk$ has a factorizing scaling limit. Then the scaling limit
   is unique. If the space $\qhcal$ is separable, all limit states
   $\uomega_0$ are regular.
\end{theorem}

\begin{proof}
It is clear that the scaling limit is unique by
Def.~\ref{def:uniqueLimit}, specializing the conditions of
Def.~\ref{def:factorLimit} to the case where $\uomega_0$ is
multiplicative, and $\zcal$ consists of a single point. Now let $\uomega_0$ be a limit state; we need to show it is regular. Let $V = \int_\zcal^{\Gamma,\oplus} d\nu(z)V_z$ be the unitary guaranteed by Def.~\ref{def:factorLimit}. By definition, the map $z \mapsto V_z \pi_z(\uA) \Omega_z$ is
measurable for any $A \in \uafk$. But we have to show that each $V_z$
fulfills the conditions of Def.~\ref{def:uniqueLimit}; in fact, we will have to
modify the $V_z$ on a null set.

First, we have $V \Omega_0 = \Omega_\zfk \otimes \qOmega$ by assumption. On the
other hand, $V \Omega_0 = \int_\zcal^{\oplus} d\nu(z)\, V_z \Omega_z $,
so that $V_z \Omega_z = \qOmega$ for $z \in \zcal\backslash\ncal_\Omega$, where
$\ncal_\Omega$ is a null set. Next we consider Poincaré transformations.
Starting from Def.~\ref{def:factorLimit}, we know that:
\begin{equation}
  V U_0(g) V\st = \idop \otimes \qU(g) \quad \text{for all } g \in \poincare.
\end{equation}
Since $U_0(g)$ factorizes by Prop.~\ref{pro:poincareDecomp}, we can
rewrite this equation as
\begin{equation}
  \int_\zcal^\oplus d\nu(z)\, V_z U_z(g) V_z\st  = \int_\zcal^\oplus d\nu(z)
  \, \qU (g).
\end{equation}
Now if $\qhcal$ is separable, and thus $L^2(\zcal,\nu,\qhcal)$ uniformly separable,
we can conclude that $V_z U_z(g) V_z\st = \qU(g)$ for all $z \in \zcal
\backslash \ncal_g$, with a null set $\ncal_g$ depending on $g$. We pick a
countable dense subset $\pcal_c$ of $\poincare$, and consider the null set
$\ncal := \ncal_\Omega \cup (\cup_{g \in \pcal_c} \ncal_g)$. Our results so far
are that
\begin{equation} \label{eqn:vZProperties}
  V_z \Omega_z = \qOmega, \quad
  V_z U_z(g) V_z\st = \qU(g)
  \quad \text{for all }
  z \in \zcal \backslash \ncal, \; g \in \pcal_c.
\end{equation}
Indeed, by continuity of the representations, the same holds for all $g \in
\poincare$. 
Now let $\hat V_z$ be those unitaries
obtained by evaluating Def.~\ref{def:factorLimit} for the multiplicative limit states $\uomega_z$. 
We set
\begin{equation}
   W_z := \begin{cases}
             V_z \quad &\text{for } z \in \zcal \backslash \ncal,
             \\
             \hat V_z & \text{for } z \in \ncal.
          \end{cases}
\end{equation}
Then we have $V =
\int_\zcal^{\Gamma,\oplus} d\nu(z) \, W_z $,  and the $W_z$ fulfill the
relations in Eq.~\eqref{eqn:vZProperties} for \emph{all} $z \in \zcal$ and $g \in \poincare$.
As a last point, $W_z \afk_z(\ocal) W_z\st \subset \qafk(\ocal)$ holds for
every $z$, since both $V_z$ and $\hat V_z$ have this property. Thus $\uomega_0$
is regular.
\end{proof}

Let us add some comments on the conditions required for $V_z$ in
Def.~\ref{def:factorLimit}, regarding Poincaré transformations and local
algebras. We could choose stricter conditions on $V_z$, requiring
that
\begin{equation}
  V_z U_z(g) V_z\st = \qU(g) \quad \text{for all }z \in \zcal \text{ and } g
  \in  \poincare.
\end{equation}
In this case, the countability problem in the proof above does not occur, and
Thm.~\ref{thm:factorToUnique} holds without the requirement that
$\qhcal$ is separable.

On the other hand, it does not seem reasonable to \emph{weaken} the conditions
on $V_z$ with respect to local algebras, requiring only that
\begin{equation} \label{eqn:vAlgWeaker}
  V \afk_0(\ocal) V\st \subset \zfk_0 \bar\otimes \qafk(\ocal) \quad
  \text{for all } \ocal.
\end{equation}
(We shall show below that this relation is implied by the chosen conditions on
$V_z$.) For if we require only \eqref{eqn:vAlgWeaker}, and we wish to apply 
the techniques used in the proof of Thm.~\ref{thm:factorToUnique}, it becomes
necessary not only to require separability of $\qhcal$ -- which seems reasonable for applications in physics --, but also separability of the algebras $\qafk(\ocal)$. That would however be too strict for our purposes, since the
local algebras are expected to be isomorphic to the hyperfinite type
$\mathrm{III}_1$ factor \cite{BDF:universal_structure}.

We now show that Eq.~\eqref{eqn:vAlgWeaker} follows from
Def.~\ref{def:factorLimit} as given.

\begin{proposition} \label{pro:algebraFactCommute}
Let $\afk$ have a factorizing scaling limit. With $V$ the unitary of
Def.~\ref{def:factorLimit}, one has $V \afk_0(\ocal) V\st \subset \zfk_0
\bar\otimes \qafk(\ocal)$ for any bounded open region $\ocal$.
\end{proposition}
\begin{proof}
Let $\uA \in \uafk(\ocal)$, and $A' \in \qafk(\ocal)'$. We compute the
commutator $[\idop \otimes A', V\pi_0(\uA)V\st]$ as a direct integral:
\begin{equation}
[\idop \otimes A', V\pi_0(\uA)V\st]
= \int_\zcal^{\oplus} d\nu(z) \,[A', V_z \pi_z(\uA) V_z\st ].
\end{equation}
Now by our requirements on the $V_z$, we have $V_z\pi_z(\uA)V_z\st \in
\qafk(\ocal)$ for all $z$, hence the commutator under the integral vanishes.
Since $\uA \in \uafk(\ocal)$ was arbitrary, this means
\begin{equation}
   V \pi_0(\uafk(\ocal))V\st \subset (\idop\otimes\qafk(\ocal)')' = \zfk_0
   \bar\otimes \qafk(\ocal).
\end{equation}
By weak closure, this inclusion extends to $V \afk_0(\ocal) V\st$.
\end{proof}

\subsection[Convergent limit => unique limit]{Convergent limit
$\Rightarrow$ unique limit}

We now assume that the theory has a convergent scaling limit, and show that our
other conditions follow. The main simplification in the convergent case is as
follows: For every $\uA \in \uafk_\mathrm{conv}$, the function $\lambda \mapsto \omega(\uA_\lambda)$
converges to a finite limit as $\lambda \to 0$; so all asymptotic means applied to this function
yield the same value. Hence the value of $\uomega_0(\uA)$ is the same for
all limit states $\uomega_0$, multiplicative or not. 

\begin{theorem} \label{thm:convergentImpliesUnique}
If the scaling limit of $\afk$ is convergent, then it is unique. 
If further a multiplicative limit state exists such that the associated limit
space $\hcal_0$ is separable, then all limit states are regular, and $\hcal_0$
is uniformly separable for \emph{any} limit state.
\end{theorem}

\begin{proof}
  We pick a fixed multiplicative limit state $\quomega$ and denote the
  corresponding representation objects as $\qhcal$, $\qpi$, $\qU$, $\qOmega$. 
  Given any other multiplicative limit state
  $\uomega_0$, we define a map $V$ by
  \begin{equation}
    V: \hcal_0 \to \qhcal, \quad \pi_0(\uA) \Omega_0 \mapsto \qpi(\uA) \qOmega
    \quad \text{for all } \uA \in \uafk_\mathrm{conv}. 
  \end{equation}
  The convergence property of $\uA \in \uafk_\mathrm{conv}$ implies
  \begin{equation} \label{eqn:uniqueUnitary}
	 \|\pi_0(\uA) \Omega_0\|^2 = \uomega_0(\uA\st\uA)
	 =\quomega(\uA\st\uA) =  \|\qpi(\uA) \qOmega\|^2,
  \end{equation}
so the linear map $V$ is both well-defined and isometric. It is also
densely defined and surjective by assumption (Def.~\ref{def:convLimit}).
Hence $V$ extends to a unitary. Using the $\ualpha$-invariance of
$\uafk_\mathrm{conv}$, one checks by direct computation that $V U_0(g)
  V\st = \qU(g)$ for all $g \in \poincare$. Also, $V \Omega_0=\qOmega$ is
  clear. Further, if $\uA \in \uafk(\ocal)\cap \uafk_\mathrm{conv}$, it is clear
  that $V \pi_0(\uA) V\st = \qpi(\uA)$. By weak density, this means $V
  \afk_0(\ocal) V\st = \afk_0(\ocal)$. Thus the scaling limit is unique.
  
  Now let $\uomega_0$ not necessarily be multiplicative. Decomposing it into
  multiplicative states $\uomega_z$ as in Prop.~\ref{pro:stateDecomp}, 
  the above construction gives us unitaries $V_z:\hcal_z \to \qhcal$ for
  every $z$. In fact, the functions $z \mapsto V_z \pi_z(\uA) \Omega_z = \qpi(\uA) \qOmega$ 
  are \emph{constant} for all $\uA \in
  \uafk_\mathrm{conv}$, in particular measurable. 
  Now let $\chi\in\qhcal$ and $\uB \in \uafk$. We can choose a sequence
  $(\uA_n)_{n \in \nbb}$ in $\uafk_\mathrm{conv}$ such that $\qpi(\uA_n)
  \qOmega\to \chi$ in norm. Noticing that
  \begin{equation}
    \hrskp{V_z \pi_z(\uB) \Omega_z}{ \chi }
    =  \lim_{n \to \infty} \hrskp{V_z \pi_z(\uB) \Omega_z}{ \qpi(\uA_n)
    \qOmega } =  \lim_{n \to \infty} \uomega_z(\uB\st \uA_n),
  \end{equation} 
  we see that the left-hand side, as a function of $z$, is the pointwise limit
  of continuous functions, and hence measurable. Thus $z \mapsto V_z \pi_z(\uB)
  \Omega_z$ is weakly measurable. Now if $\qhcal$ was chosen separable, which
  is possible by assumption, weak measurability implies Lusin measurability
  of the function (cf. Appendix). Thus $\uomega_0$ is regular.
  
  Finally, in the separable case, we remark that we can pick a countable subset
  of $\uafk_\mathrm{count} \subset \uafk_\mathrm{conv}$ such that
  $\qpi(\uafk_\mathrm{count}) \qOmega$ is dense in $\qhcal$.
  Then $\pi_0(\uafk_\mathrm{count})\Omega_0$ becomes a fundamental sequence in
  $\hcal_0$, so that this space is uniformly separable.
\end{proof}

Of course, it follows as a corollary to the preceding sections that the limit is
also factorizing. Let us spell this out more explicitly.

\begin{proposition} \label{pro:convToFactorizing}
Suppose that $\afk$ has a convergent scaling limit, and that there exists a
multiplicative limit state $\quomega$ for which the representation space
$\qhcal$ is separable. Let $\uomega_0$ be any scaling limit state. There exists
a unitary $V = \int_\zcal^{\Gamma,\oplus} d\nu(z)\, V_z : \hcal_0 \to L^2(\zcal,\nu,\qhcal)$
such that
\begin{equation*}
  V \pi_0(\uA\,\uC) \Omega_0 = \pi_0(\uC) \Omega_\zfk \otimes \qpi(\uA) \qOmega
  \quad \text{for all } \uA \in \uafk_\mathrm{conv}, \, \uC \in \zfk(\uafk),
\end{equation*} 
and such that the $V_z$ fulfill all requirements of Def.~\ref{def:factorLimit}.
\end{proposition}

\begin{proof}
We use notation as in the proof of Thm.~\ref{thm:convergentImpliesUnique}.
Let $V_z:\hcal_z \to \qhcal$ be the unitaries constructed there. 
Then, $z \mapsto V_z\st$ is a measurable family of operators.
Namely, for any $\uA \in \uafk_\mathrm{conv}$, we find
\begin{equation}
 V_z\st \qpi(\uA) \qOmega = \pi_z(\uA) \Omega_z
\end{equation}
which is in $\Gamma$; hence measurability is checked on the
fundamental family (cf.~Lemma~\ref{lem:measOnGamma}). So the operator 
\begin{equation}
  V\st := \int_\zcal^{\oplus,\Gamma} d\nu(z) \, V_z\st
\end{equation}
is well-defined. Domain and range of $V\st$ are both uniformly separable, see
Thm.~\ref{thm:convergentImpliesUnique}. Thus also the adjoint of $V\st$,
denoted as $V$, is decomposable with $V = \int_\zcal^{\Gamma,\oplus} d\nu(z)\, V_z $. 
It is then clear that $V$ is unitary. Also, we have for $\uA \in \uafk_\mathrm{conv}$ and
$\uC \in \zfk(\uafk)$,
\begin{multline}
  V \pi_0(\uA\, \uC) \Omega_0 
  = V \int_\zcal^\Gamma d\nu(z) \, \pi_z(\uC) \pi_z(\uA) \Omega_z
  = \int_\zcal^\oplus d\nu(z)\, \pi_z(\uC) V_z \pi_z(\uA) \Omega_z
  \\
  = \int_\zcal^\oplus d\nu(z)\, \pi_z(\uC) \qpi(\uA) \qOmega
  = (\pi_0(\uC)  \Omega_\zfk) \otimes ( \qpi(\uA) \qOmega).
\end{multline}
As a direct consequence of the discussion following
Eq.~\eqref{eqn:uniqueUnitary}, the $V_z$ have all the properties 
required in Def.~\ref{def:factorLimit} regarding vacuum vector, symmetries, and
local algebras.
\end{proof}

\section{Dilation covariance in the limit} \label{sec:dilations}

Our next aim is to analyze the structure of dilation symmetries in the limit
theory. To that end, we consider a scaling limit state $\uomega_0$ which
is invariant under $\udelta_\mu$. 
As shown in~\cite[sec.\ 2]{BDM:scaling_fields}, the associated limit
theory is covariant with respect to a strongly continuous unitary representation
$g \in \ugcal \mapsto U_0(g)$ of the extended symmetry group $\ugcal$,
including both Poincar\'e symmetries and dilations. Our interest is how the
dilation unitaries $U_0(\mu)$ relate to decomposition theory in
Sec.~\ref{sec:decompTheory}, and how they behave in the more specific
situations analyzed in Sec.~\ref{sec:uniqueFactorLim}. We
will consider three cases of decreasing scope: first, the general situation; 
second, the factorizing scaling limit; third, the convergent scaling limit.

We first consider a general theory as in Sec.~\ref{sec:decompTheory}, and
analyze the decomposition of the dilation operators corresponding to the direct
integral decomposition of $\hcal_0$ introduced in Thm.~\ref{thm:hilbertDecomp}. To
this end, we first note that $\udelta_\mu$ leaves $\zfk(\uafk)$ invariant;
thus we have a representation of the dilations $U_\zfk(\mu) := U_0(\mu)
\restrict \hcal_\zfk$ on $\hcal_\zfk$. Identifying $\hcal_\zfk$ with $L^2(\zcal,\nu)$
as before, the $U_\zfk(\mu)$ act on a function space.
This action, and its extension to the entire Hilbert space, can be described in more
detail.

\begin{proposition}\label{prop:dilationDecomp}
Let $\uomega_0$ be an invariant limit state. There exist an
action of the dilations through homeomorphisms $z \mapsto \mu.z$ of $\zcal$, 
and unitary operators $U_z(\mu) : \hcal_z \to
\hcal_{\mu. z}$ for $\mu \in \rbb_+$, $z \in \zcal$, such that:
\begin{enumerate}
\localitemlabels
\item \label{it:measureInvar}the measure $\nu$ is invariant under
the transformation $z\mapsto \mu.z$;
\item \label{it:dilationHomeo}$\big(U_\zfk(\mu)\chi\big)(z) =
\chi(\mu^{-1}.z)$ for all $\chi \in L^2(\zcal,\nu)$, as an equation in the
$L^2$ sense;
\item \label{it:UzProperties}$U_z(1) = \idop$, $U_z(\mu)\st = U_{\mu.
z}(\mu^{-1})$, 
$U_{\mu.z}(\mu')U_z(\mu) = U_z(\mu'\mu)$ for all $z \in \zcal$,
$\mu,\mu' \in \rbb_+$;
\item \label{it:UzPoincare}
$U_z(\mu) U_z(x,\Lambda) = U_{\mu.z}(\mu x,\Lambda)U_z(\mu)$
for all $z \in \zcal$, $(x,\Lambda) \in \poincare$;
\item \label{it:dilationDecomp}
$U_0(\mu)\chi = \int_\zcal^\Gamma d\nu(z)\,U_{\mu^{-1}.z}(\mu)
\chi(\mu^{-1}.z)$ for all $\chi \in \int_\zcal^\Gamma d\nu(z)\,\hcal_z$.
\end{enumerate}
\end{proposition}

\begin{proof}
Recalling that $\zcal$ is the spectrum of the commutative
\cistar{} algebra $\pi_0(\zfk(\uafk))$, we define the homeomorphism $z \mapsto
\mu.z$ as the one induced by the automorphism $\ad U_\zfk(\mu^{-1})$ of $\pi_0(\zfk(\uafk))$. 
For $\uC \in \zfk(\uafk)$, we know that $\pi_0(\uC)\Omega_0 \in \hcal_\zfk$
corresponds to the function $\chi_{\uC}(z)= \uomega_z(\uC)$, precisely
the image of $\pi_0(\uC)$ in the Gelfand isomorphism. Applying $U_0(\mu^{-1})$
to this vector, one obtains 
\begin{equation} \label{eqn:chiUC}
 (U_\zfk(\mu^{-1}) \chi_{\uC})(z)=\chi_{\uC}(\mu.z);
\end{equation}
thus~\ref{it:dilationHomeo} holds for all $\chi \in \ccal(\zcal)$.
Taking the scalar product of Eq.~\eqref{eqn:chiUC} with $\Omega_0$,
one sees that $\int_\zcal d\nu(z) \chi(z) = \int_\zcal d\nu(z)
\chi(\mu.z)$ for all $\mu$ and $\chi \in \ccal(\zcal)$, so \ref{it:measureInvar} follows.  
Now for general $\chi \in L^2(\zcal,\nu)$, statement~\ref{it:dilationHomeo}
follows by density.

Translating the action of $z \mapsto \mu.z$ to the level of algebras, it is
easy to see that
\begin{equation}
\uomega_{\mu.z} \circ \udelta_\mu \restrict \zfk(\uafk)
= \uomega_z \restrict \zfk(\uafk).
\end{equation}
Since however $\udelta_\mu$ commutes with the projector
$\omega:\uafk\to\zfk(\uafk)$, the same equation holds on all of $\uafk$.
Therefore, the maps $U_z(\mu) :\hcal_z \to \hcal_{\mu.z}$ given by
\begin{equation}
U_z(\mu)\pi_z(\uA)\Omega_z := \pi_{\mu. z}(\udelta_\mu(\uA))\Omega_{\mu. z}
\end{equation}
are well-defined and unitary. The properties of $U_z(\mu)$ listed
in~\ref{it:UzProperties} and~\ref{it:UzPoincare} then follow from this definition by easy
computations.

Now for~\ref{it:dilationDecomp}: As before, we identify $\hcal_0$ with
$\int_\zcal^\Gamma d\nu(z)\,\hcal_z$. Then we have, for all $\uA \in \uafk$,
\begin{equation}
\begin{split}
U_0(\mu)\pi_0(\uA)\Omega_0 &= \pi_0(\udelta_\mu(\uA))\Omega_0 
   = \int_\zcal^\Gamma d\nu(z)\,\pi_z(\udelta_\mu(\uA))\Omega_z\\
   &= \int_\zcal^\Gamma
    d\nu(z)\, U_{\mu^{-1}.z}(\mu)\pi_{\mu^{-1}. z}(\uA)\Omega_{\mu^{-1}. z}.
\end{split}
\end{equation}
Given now a vector $\chi \in \int_\zcal^\Gamma d\nu(z)\,\hcal_z$, we can find a
sequence $(\pi_0(\uA_n)\Omega_0)_{n \in \nbb}$ converging in norm to $\chi$. 
Passing to a subsequence, we can also assume that $\pi_z(\uA_n)\Omega_z \to
\chi(z)$ in norm for almost every $z \in \zcal$. Hence, using the dominated
convergence theorem and \ref{it:measureInvar}, we see that
\begin{equation*}
\lim_{n\to+\infty}\int_\zcal^\Gamma d\nu(z)\, U_{\mu^{-1}.z}(\mu)
\pi_{\mu^{-1}.z}(\uA_n)\Omega_{\mu^{-1}.z} = \int_\zcal^\Gamma
d\nu(z)\, U_{\mu^{-1}. z}(\mu)\chi(\mu^{-1}. z),
\end{equation*}
which gives~\ref{it:dilationDecomp}.
\end{proof} 

Thus dilations act between the fibers of the direct integral decomposition by
unitaries $U_z(\mu)$, which depend on the fiber. They fulfill the cocycle-type
composition rule $U_{\mu.z}(\mu')U_z(\mu) = U_z(\mu'\mu)$ that one would
naively expect; cf.~also the theory of equivariant disintegrations
for separable $C\st$ algebras \cite[Ch.~X §3]{Tak:TOA2}.

We shall now further restrict to the situation of a factorizing scaling limit,
as in Def.~\ref{def:factorLimit}, in which the fiber spaces $\hcal_z$ are all identified with a unique space
$\qhcal$. By this identification, we can regard the unitaries $U_z(\mu)$ as
endomorphisms $\hat U_z(\mu)$ of $\qhcal$. Our result for these endomorphisms
is as follows.

\begin{proposition}
Let $\uomega_0$ be an invariant limit state. Suppose that the scaling
limit of $\uafk$ is factorizing, and let $V=\int d\nu(z) V_z$ be the unitary of
Def.~\ref{def:factorLimit}. Then, the unitary operators
\begin{equation*}
    \hat U_z(\mu): \qhcal \to \qhcal, \quad
    \hat U_z(\mu) = V_{\mu.z} U_z(\mu) V_z\st
\end{equation*}
fulfill for any $z \in \zcal$, $\mu,\mu' \in \rbb_+$
the relations
\begin{equation*}
\hat U_z(1) = \idop, 
\quad 
\hat U_z(\mu)\st = \hat U_{\mu.z}(\mu^{-1}), 
\quad
\hat U_{\mu.z}(\mu') \hat U_z(\mu) = \hat U_z(\mu'\mu).
\end{equation*}
If $\hcal_0$ is uniformly separable, one has
\begin{equation*}
  V U_0(\mu) V\st = (U_\zfk(\mu) \otimes \idop) \int_\zcal^{\oplus} d\nu(z)
  \hat U_z(\mu);
\end{equation*}
and for every $\mu>0$, there is a null set $\ncal
\subset\zcal$ such that for any $(x,\Lambda)\in\poincare$ and any $z \in
\zcal\backslash\ncal$,
\begin{equation*}
   \hat U_z(\mu) \qU(x,\Lambda) = \qU(\mu x,\Lambda) \hat U_z(\mu).
\end{equation*}
\end{proposition}

\begin{proof}
It is clear that $\hat U_z(\mu)$, defined as above, are unitary, and their
composition relations follow from Prop.~\ref{prop:dilationDecomp}
\ref{it:UzProperties}. Now let $\hcal_0$ be uniformly separable. Then, together with
$V$, also $V\st$ is decomposable. By a short computation, one finds for any
$\chi \in L^2(\zcal,\nu,\qhcal)$:
\begin{equation} \label{eqn:uZeroTransf}
  V U_0(\mu) V\st \chi =
  \int_\zcal^{\oplus} d\nu(z) \, V_z U_{\mu^{-1}.z}(\mu) V_{\mu^{-1}.z}\st \,
  \chi(\mu^{-1}.z).
\end{equation}
Now, following Prop.~\ref{prop:dilationDecomp} \ref{it:dilationHomeo}, the
operator $U_\zfk(\mu) \otimes \idop$ acts on vectors $\chi'$ via
\begin{equation}
  (U_\zfk(\mu) \otimes \idop ) \chi' =
  \int_\zcal^{\oplus} d\nu(z) \, \chi'(\mu^{-1}.z). 
\end{equation}
Together with Eq.~\eqref{eqn:uZeroTransf}, this entails
\begin{equation} \label{eqn:dilFactor}
  V U_0(\mu) V\st \chi =
  (U_\zfk(\mu) \otimes \idop )  \int_\zcal^{\oplus} d\nu(z) \, V_{\mu.z}
  U_{z}(\mu) V_{z}\st \, \chi(z),
\end{equation}
of which the second assertion follows. Further, one computes from $V
U_0(x,\Lambda) V\st = \idop \otimes \qU(x,\Lambda)$ and from
Eq.~\eqref{eqn:dilFactor} that 
\begin{equation}
  \int_\zcal^\oplus d\nu(z) \hat U_z(\mu) \qU(x,\Lambda)
  = V U_0(\mu,x,\Lambda) V\st
  = \int_\zcal^\oplus d\nu(z) \qU(\mu x,\Lambda) \hat U_z(\mu).
\end{equation}
Uniform separability implies that the integrands agree except on a null
set. This null set may depend on $x,\Lambda$. However, we can choose it uniformly on
a countable dense set of the group, and hence, by continuity, uniformly for all
group elements.
\end{proof}

This shows that the dilation symmetries factorize into a central part,
$U_\zfk(\mu)\otimes 1$, which ``mixes'' the fibers of the direct integral,
and a decomposable part, $\int_\zcal^{\oplus} d\nu(z) \hat U_z(\mu)$. The
unitaries $\hat U_z(\mu)$ will generally depend on $z$; and like the
$U_z(\mu)$ before, they do not necessarily fulfill a group relation, but a
cocycle equation
\begin{equation} \label{eqn:cocycle}
   \hat U_{\mu.z}(\mu') \hat U_z(\mu) = \hat U_z(\mu'\mu),
\end{equation}
as shown above, where $\mu.z$ can in general not be replaced with $z$.
However, using the commutation relations with the other parts of the
symmetry group, one sees that $\hat U_{z}(\mu') \hat U_z(\mu)  \hat
U_z(\mu'\mu)\st$ is (a.~e.) an inner symmetry of the theory $\qafk$. 
On the other hand, this representation property ``up to an inner symmetry''
cannot be avoided if such symmetries exist in the theory at all; 
for they might be multiplied to $V_z$ in a virtually arbitrary fashion at any
point $z$. In this respect, we encounter a similar situation with respect to dilation symmetries as Buchholz and Verch
\cite{BucVer:scaling_algebras}. In the present context, however, it
seems more transparent how this cocycle arises.

Under somewhat stricter assumptions, we can prove a stronger result that avoids
the ambiguities discussed above. Let us consider the case of a \emph{convergent}
scaling limit, per Def.~\ref{def:convLimit}. In this case, we shall see that the $U_z(\mu)$ can
actually be chosen independent of $z$, and yield a group representation in
the usual sense.

\begin{proposition} \label{pro:ConvergentDilation}
Let $\afk$ have a convergent scaling limit, and let $\quomega$ be a
multiplicative limit state with separable representation space $\qhcal$. Then
the Poincaré group representation $\qU$ on $\qhcal$ extends to a representation
of the extended symmetry group $\ugcal$.
For any invariant limit state $\uomega_0$ with associated representation $U_0$
of $\ugcal$, one has
\begin{equation*}
V U_0(\mu) V\st = (U_0(\mu)\restrict \hcal_\zfk) \otimes \qU(\mu),
\end{equation*}
where $V$ is the unitary introduced in Proposition \ref{pro:convToFactorizing}.
\end{proposition}

\begin{proof}
With $\quomega$, also every $\quomega \circ \udelta_\mu$ is a
scaling limit state. 
Thanks to the invariance of $\uafk_\mathrm{conv}$ under dilations,
we thus have for each $\uA \in \uafk_\mathrm{conv}$, 
\begin{equation}
\| \qpi(\udelta_\mu(\uA))\qOmega \|^2 = \quomega\circ\udelta_\mu(\uA^*\uA) =
\quomega(\uA^*\uA) = \| \qpi(\uA)\qOmega \|^2.
\end{equation} 
This yields the existence of a unitary strongly continuous representation
$\mu \mapsto \qU(\mu)$ on $\qhcal$ such that
\begin{equation}
\qU(\mu)\qpi(\uA)\qOmega = \qpi(\udelta_\mu(\uA))\qOmega, \qquad \uA \in \uafk_\mathrm{conv}.
\end{equation}
That also implies
\begin{equation}
\qU(\mu)\qU(x,\Lambda)\qpi(\uA)\qOmega =
\qpi(\ualpha_{\mu,x,\Lambda}(\uA))\qOmega, \qquad \uA \in \uafk_\mathrm{conv},
\end{equation}
which shows that $(\mu,\Lambda,x) \mapsto \qU(\mu)\qU(\Lambda,x)$ is a unitary
representation of $\ugcal$ on $\qhcal$, extending the representation of the
Poincar\'e group.

Now if $V:\hcal_0 \to \hcal_{\zfk}\otimes\qhcal$ is the unitary of
Prop.~\ref{pro:convToFactorizing}, a calculation shows that
\begin{equation}
VU_0(\mu)V\st \big( \pi_0(\uC)\Omega_0 \otimes \qpi(\uA)\qOmega \big)=
\pi_0(\udelta_\mu(\uC))\Omega_0\otimes\qpi(\udelta_\mu(\uA))\qOmega, 
\quad \uC \in \zfk(\uafk), \uA \in \uafk_\mathrm{conv},
\end{equation}
which entails that $VU_0(\mu)V^* = (U_0(\mu)\restrict \hcal_\zfk) \otimes \qU(\mu)$.
\end{proof}

Thus, the limit theory is 
``dilation covariant'' in the usual sense, with a unitary acting on $\qhcal$.
Considering the unitaries $\idop \otimes \qU(g)$, we actually get a unitary
representation in \emph{any} limit theory, even corresponding to
multiplicative states. Only for compatibility with the scaling limit
representation $\pi_0$ it is necessary to consider invariant means, and to take
$U_0(\mu)\restrict \hcal_\zfk$ into account.

\section{Phase space properties} \label{sec:compactness}

In this section, we wish to investigate how the notion of phase space
conditions, specifically the (quite weak) Haag-Swieca compactness condition
\cite{HaaSwi:compactness}, fits into our context, and how it transfers to the
limit theory. An important aspect here is that Haag-Swieca compactness of a
quantum field theory guarantees that the corresponding Hilbert space is
separable; this property transfers to \emph{multiplicative} limit states
in certain circumstances \cite{Buc:phase_space_scaling}. We shall give a
strengthened version of the compactness condition that guarantees
our general limit spaces to be uniformly separable, a property that turned out to be
valuable in the previous sections.

We need some extra structures to that end. First, we consider
``properly rescaled'' vector-valued functions $\uchi: \rbb_+ \to \hcal$.
Specifically, for $\uA \in \uafk$, let $\uA\Omega$ denote the function $\lambda \mapsto
\uA_\lambda\Omega$. We set
\begin{equation}
  \uhcal = \clos \{ \uA \Omega \,|\, \uA \in \uafk \},
\end{equation}
where the closure is taken in the supremum norm $\|\uchi\| =
\sup_\lambda \|\uchi_\lambda\|$. Then $\uhcal$ is a Banach space, in fact a
Banach module over $\zfk(\uafk)$ in a natural way. Given a limit state, we
transfer the limit representation $\pi_0$ to vector-valued functions. To that
end, consider the space $\ccal(\Gamma)$ of $\Gamma$-continuous vector fields, as
defined in the Appendix. We define $\eta_0: \uhcal \to \ccal(\Gamma)$ on a
dense set by
\begin{equation}
  \eta_0 (\uA \Omega ) := \pi_0(\uA) \Omega_0.
\end{equation} 
This is well-defined, since one computes
\begin{multline}
  \|\pi_0(\uA) \Omega_0\|_\infty = 
  \sup_{z \in \zcal} \|\pi_z(\uA) \Omega_z\|
  = \big( \sup_{z \in \zcal} \uomega_z(\uA\st\uA) \big)^{1/2}
  \\
  \leq \| \omega(\uA\st\uA) \|^{1/2} 
  = \big( \sup_{\lambda > 0} \|\uA_\lambda \Omega\|^2 \big)^{1/2}
  = \|\uA \Omega\|.
\end{multline}
That also shows $\|\eta_0\| \leq 1$. Note that $\eta_0$ fulfills
\begin{equation} \label{eqn:etaModuleHom}
  \eta_0( \uC \uchi ) = \pi_0(\uC) \eta_0(\uchi) \quad
  \text{for all } \uC \in \zfk(\uafk), \uchi \in \uhcal,
\end{equation}
this easily being checked for $\uchi = \uA\Omega$. So $\eta_0$ preserves the
module structure in this sense. Further,  $\eta_0 : \uhcal \to
\ccal(\Gamma)$ clearly has dense range.

It is important in our context that $\uhcal$ is left invariant under
multiplication with suitably rescaled functions of the Hamiltonian. More
precisely, we denote these functions as $f(\uH)$ for $f \in \scal(\rbb_+)$;
they are defined as elements of $\ubfk$ by  $f(\uH)_\lambda = f(\lambda H)$,
with norm $\|f(\uH)\|\leq \|f\|_\infty$. They act on $\uhcal$ by pointwise
multiplication. The following lemma generalizes an observation in
\cite{Buc:phase_space_scaling}.

\begin{lemma} \label{lem:fBeta}
  Let $f \in \scal(\rbb_+)$. Then, for each $\uchi \in \uhcal$, we have 
  $f(\uH) \uchi \in \uhcal$. There exists a test function $g \in \scal(\rbb)$
  such that for all $\uA \in \uafk$,
  \begin{equation*}
  	f(\uH) \uA \Omega = \ualpha_g \uA \Omega := \big(\int dt \,g(t)\, \ualpha_t
  	\uA\big) \Omega.
  \end{equation*} 
\end{lemma}
\begin{proof}
We continue $f$ to a test function $\hat f \in \scal(\rbb)$, and choose $g$ as
the Fourier transform of $\hat f$. One finds by spectral analysis of $H$
that for any $\uA \in \uafk$,
\begin{equation}
   f(\lambda H) \uA_\lambda \Omega = \int_0^\infty \hat f(\lambda E) dP(E)
   \uA_\lambda \Omega = \int dt\, g (t) e^{\imath \lambda H t} \uA_\lambda
   \Omega = (\ualpha_g \uA)_\lambda \Omega .
\end{equation}
This shows that $f(\uH) \uA \Omega$ has the proposed form, and is an element of
$\uhcal$. Since $\|f(\lambda H)\| \leq \|f\|_\infty$ uniformly in $\lambda$,
we may pass to limits in $\uA\Omega$ and obtain that $f(\uH) \uchi \in
\uhcal$ for all $\uchi \in \hcal$.
\end{proof}

As a next step towards phase space conditions, let us explain a notion of
compact maps adapted to our context. To that end, let $\ecal$ be a Banach space
and $\fcal$ a Banach module over the commutative Banach algebra $\rcal$. We say
that a linear map $\psi:\ecal \to \fcal$ is \emph{of uniform rank~1} if it
is of the form $\psi = e(\cdotarg) f$ with $e: \ecal \to \rcal$ linear and
continuous, and $f \in \fcal$. Sums of $n$ such terms are called \emph{of
uniform rank $n$.}\footnote{Note that the ``uniform rank'' is rather an
upper estimate, in the sense that a map of uniform rank $n$ may at the
same time be of uniform rank $n-1$.} 
We say that $\psi$ is \emph{uniformly compact} if it is an infinite sum of
terms of uniform rank 1, $\psi = \sum_{j=0}^\infty e_j(\cdotarg) f_j$,
where the sum converges in the Banach norm. For $\rcal=\cbb$, these definitions reduce to the usual notions
of compact or finite-rank maps.

We are now in the position to consider Haag-Swieca compactness. We fix, once
and for all, an element $\cutoff \in \zfk(\uafk)$ with $\|\cutoff\|\leq 1$,
$\cutoff_\lambda=0$ for $\lambda>1$, and $\cutoff_\lambda=1$ for $\lambda <
1/2$. For a given $\beta>0$ and any bounded region $\ocal$, we consider
the map
\begin{equation}
  \uTheta^{(\beta,\ocal)}: \uafk(\ocal) \to \uhcal,\quad
  \uA \mapsto e^{-\beta \uH} \,\cutoff \uA \Omega.
\end{equation}
This is indeed well-defined due to Lemma~\ref{lem:fBeta}. Our variant of
the Haag-Swieca compactness condition, uniform at small scales, is then as
follows.
\begin{definition} \label{def:uniformCompactness}
A quantum field theory fulfills the \emph{uniform Haag-Swieca
compactness condition} if, for each bounded region $\ocal$, there is $\beta>0$
such that the map $\uTheta^{(\beta,\ocal)}$ is uniformly compact.
\end{definition}
We note that this property is independent of the choice of $\cutoff$; the
role of that factor is to ensure that we restrict our attention to the
short-distance rather than the long-distance regime. We do not discuss
relations of uniform Haag-Swieca compactness with other versions of phase space
conditions here. Rather, we show in Sec.~\ref{sec:examples} that the
condition is fulfilled in some simple models.

We now investigate how the compactness property
transfers to the scaling limit. 
To that end, we consider the corresponding phase space map in
the limit theory,
\begin{equation}
  \Theta_0^{(\beta,\ocal)}:
  \afk_0(\ocal) \to \hcal_0,
  \quad
  A \mapsto e^{-\beta H_0} A \Omega_0.
\end{equation}
Its relation to $\uTheta^{(\beta,\ocal)}$ is rather direct.

\begin{proposition} \label{pro:thetaToLimit}
For any fixed $\ocal$ and $\beta>0$, one has $\eta_0 \circ
\uTheta^{(\beta,\ocal)} = \Theta_0^{(\beta,\ocal)} \circ \pi_0$. If
$\uTheta^{(\beta,\ocal)}$ is uniformly compact, so is $\Theta_0^{(\beta,\ocal)}
\circ \pi_0$.
\end{proposition}

\begin{proof}
Given $\beta$, we choose a function $g_\beta$ relating to
$f_\beta(E)=\exp(-\beta E)$ per Lemma~\ref{lem:fBeta}. For any
$\uA \in \uafk(\ocal)$, we compute
\begin{equation}
   \eta_0 \uTheta^{(\beta,\ocal)}(\uA) 
   = \eta_0( \cutoff \ualpha_{g_\beta} \uA \Omega)
   = \pi_0(\cutoff) \pi_0(\ualpha_{g_\beta} \uA) \Omega_0
   = \alpha_{0,g_\beta} \pi_0(\uA) \Omega_0
   = \Theta_0^{(\beta,\ocal)} \pi_0(\uA).
\end{equation}
Thus $\eta_0 \circ
\uTheta^{(\beta,\ocal)} = \Theta_0^{(\beta,\ocal)} \circ \pi_0$ as proposed.
Now let $\uTheta^{(\beta,\ocal)}$ be uniformly compact, $\uTheta^{(\beta,\ocal)} =
\sum_j e_j(\cdotarg) f_j$. Then $\eta_0$ can be exchanged with the infinite
sum due to continuity, which yields
\begin{equation} \label{eqn:theta0Compact}
  \Theta_0^{(\beta,\ocal)} \circ \pi_0
  = \sum_j \eta_0(e_j(\cdotarg) f_j)
  = \sum_j (\pi_0 \circ e_j(\cdotarg)) (\eta_0 f_j),
\end{equation}
using Eq.~\eqref{eqn:etaModuleHom}. Thus $\Theta_0^{(\beta,\ocal)} \circ \pi_0$ 
is uniformly compact.
\end{proof}

The above results show in particular that $\img \Theta_0^{(\beta,\ocal)} \circ
\pi_0 \subset \ccal(\Gamma)$. Since we can write
\begin{equation}
  \Theta_0^{(\beta,\ocal)} \circ \pi_0(\uA)
  = \int d\nu(z) \, \Theta_z^{(\beta,\ocal)} \circ \pi_z(\uA)
\end{equation}
with the obvious definition of $\Theta_z^{(\beta,\ocal)}$, the above
proposition establishes a rather strong form of compactness in
the limit theory, uniform in $z$; note that the sum in
Eq.~\eqref{eqn:theta0Compact} converges with respect to the supremum norm. 

We now come to the main result of the section, showing that compactness in the
above form implies uniform separability of the limit Hilbert space.
\begin{theorem} \label{thm:compactSeparable}
Suppose that the theory $\afk$ fulfils uniform Haag-Swieca compactness. Then,
for any limit state $\uomega_0$, the representation space $\hcal_0$ is
uniformly separable, where the fundamental sequence can be chosen from
$\ccal(\Gamma)$.
\end{theorem}

\begin{proof} 
We choose a sequence of regions $\ocal_k$ such that $\ocal_k
\nearrow \mkraum$, and a sequence $(\beta_k)_{k\in\nbb}$ in $\rbb_+$ such
that all $\uTheta^{(\beta_k,\ocal_k)}$ are uniformly compact. 
By Prop.~\ref{pro:thetaToLimit} above, also $\Theta_0^{(\beta_k,\ocal_k)}\circ\pi_0$ are uniformly compact. Explicitly, choose $e_j^{(k)}: \uafk(\ocal_k) \to \ccal(\zcal)$ and $f_j^{(k)} \in \ccal(\Gamma)$
such that
\begin{equation}
  \Theta_0^{(\beta_k,\ocal_k)}\circ\pi_0    
  = \sum_j e_j^{(k)}(\cdotarg) f_j^{(k)}.
\end{equation}
We will construct a fundamental sequence using the $f_j^{(k)}$.
To that end, let $\uA \in \uafk(\ocal)$ for some $\ocal$. For $k$ large enough,
we know that
\begin{equation} 
  e^{-\beta H_0} \pi_0(\uA) \Omega_0 
  = \Theta_0^{(\beta_k,\ocal_k)} (\pi_0(\uA))
  = \sum_j e_j^{(k)}(\uA) f_j^{(k)}.
\end{equation}
The sum converges in the supremum norm, i.e., uniformly at all points $z$. 
Let us choose a fixed $z$. Then, it is clear that 
\begin{equation} 
  e^{- H_z^2} \pi_z(\uA) \Omega_z 
  = \sum_j \big(e_j^{(k)}(\uA)\big)(z)\,    e^{- H_z^2+\beta_k H_z}
  f_j^{(k)}(z),
\end{equation}
noting that $\exp(-H_z^2+\beta_k H_z)$ is a bounded operator. Observe that
$(e_j^{(k)}(\uA))(z)$ are merely numerical factors.
Since $\uA$ and $\ocal$ were arbitrary, and $\cup_\ocal
\pi_z(\uafk(\ocal))\Omega_z$ is dense in $\hcal_z$, this means
\begin{equation}
  e^{- H_z^2} \hcal_z \subset 
  \clos\, \lspan \{\,e^{- H_z^2+\beta_k H_z} f_j^{(k)}(z) |\, j,k\in\nbb\}.
\end{equation}
Now $\exp(- H_z^2) $ is a selfadjoint operator with trivial kernel, thus its
image is dense.  Hence the $\exp(- H_z^2+\beta_k H_z) f_j^{(k)}(z)$ are total
in $\hcal_z$. This holds for all $z$, thus 
$\{ \exp(- H_0^2+\beta_k H_0) f_j^{(k)}\,|\,j,k\in\nbb\}$ is a fundamental
sequence. Applying Lemma~\ref{lem:fBeta} to $f(E)=\exp(- E^2+\beta_k E)$, we
find that the elements of the fundamental sequence lie in $\ccal(\Gamma)$.
\end{proof}

\section{Examples} \label{sec:examples}

We are now going to investigate the structures discussed in simple models.
Particularly, we wish to show that our conditions on ``convergent scaling
limits'' (Def.~\ref{def:convLimit}) and ``uniform Haag-Swieca compactness''
(Def.~\ref{def:uniformCompactness}) can be fulfilled at least in simple
situations. To that end, we first consider the situation where the
theory $\afk$ ``at finite scales'' is equipped with a dilation symmetry.
Then, we investigate the real scalar free field as a concrete example.

\subsection{Dilation covariant theories}
We now consider the case where the net $\afk$, which our investigation
starts from, is already dilation covariant. One expects that the scaling limit
construction reproduces the theory $\afk$ in this case, and that the dilation
symmetry obtained from the scaling algebra coincides with the original one. We
shall show that this is indeed the case under a mild phase space condition,
and also that this implies the stronger phase space condition in
Def.~\ref{def:uniformCompactness}. This extends a discussion in
\cite[Sec.~5]{BucVer:scaling_algebras}. 

Technically, we will assume in the following that $\afk$ is a local net in the
vacuum sector with symmetry group $\gcal$, which is generated by the
Poincaré group and the dilation group. We shall denote the corresponding
unitaries as $U(\mu,x,\Lambda)=U(\mu)U(x,\Lambda)$. The mild phase space condition referred to is the Haag-Swieca compactness condition
for the original theory: We assume that for each bounded region $\ocal$ in
Minkowski space, there exists $\beta>0$ such that the map
$\Theta^{(\beta,\ocal)}:\afk(\ocal) \to \hcal$, $A \mapsto \exp(-\beta H) A
\Omega$ is compact. (This is equivalent to a formulation where the factor
$\exp(-\beta H)$ is replaced with a sharp energy cutoff, as used in
\cite{HaaSwi:compactness}.)

\begin{theorem}\label{thm:dilation}
Let $\afk$ be a dilation covariant net in the vacuum sector which satisfies the Haag-Swieca compactness condition. Then $\afk$ has a convergent scaling limit.
\end{theorem}

\begin{proof}
We introduce the \cistar-subalgebra $\hat{\afk}(\ocal) \subset \afk(\ocal)$ of
those elements $A \in \afk(\ocal)$ for which $g \mapsto \alpha_g(A)$ is
norm continuous. Since the symmetries are implemented by
continuous unitary groups, $\hat\afk(\ocal)$ is strongly dense in
$\afk(\ocal)$. We then define a \cistar-subalgebra of the scaling algebra
$\uafk(\ocal)$,
\begin{equation}
\uafk_\mathrm{conv}(\ocal) := \{ \lambda \mapsto U(\lambda)AU(\lambda)^*\,|\, A \in \hat\afk(\ocal)\},
\end{equation}
and the $\ualpha$-invariant algebra $\uafk_\mathrm{conv}\subset \uafk$ is defined as the \cistar-inductive limit of the $\uafk_\mathrm{conv}(\ocal)$.

It is evident that condition (i) in Def.~\ref{def:convLimit} is
fulfilled by $\uafk_\mathrm{conv}$, as the functions $\lambda \mapsto
\omega(\uA_\lambda)$ are constant in the present case. Now let $\uomega_0$ be
a multiplicative limit state. With similar arguments\footnote{%
Since in contrast to \cite{BucVer:scaling_algebras}, we here take
the $\afk_0(\ocal)$ to be \wstar{} algebras, we need to amend the argument in
step (d) of \cite[Prop.~5.1]{BucVer:scaling_algebras} slightly: We first construct the isomorphism $\phi$ on the \cistar{} algebra
$\pi_0(\uafk(\ocal))$, and then continue it to the weak closure;
cf.~\cite[Lemma~10.1.10]{KadRin:algebras2}.}
as in~\cite[Prop.~5.1]{BucVer:scaling_algebras}, using the Haag-Swieca
compactness condition, we can construct a net isomorphism $\phi$ from $\afk_0$
to $\afk$, which has the property that if $\uA_\lambda = U(\lambda) A U(\lambda)^*$ with $A \in \hat\afk(\ocal)$, then
$\phi(\pi_0(\uA)) = A$. From this, and from the strong density of
$\hat\afk(\ocal)$ in $\afk(\ocal)$, it follows that $\pi_0(\uafk_\mathrm{conv}(\ocal))$ is strongly dense in $\afk_0(\ocal)$. 
Thus condition (ii) in Def.~\ref{def:convLimit} is satisfied as well.
\end{proof}

Since the isomorphism $\phi$ above can be shown to intertwine the
respective vacuum states, it is actually the adjoint action of a unitary $W:
\hcal_0 \to \hcal$. We remark that $\hcal$, and then also $\hcal_0$, is
separable due to the Haag-Swieca compactness condition. Then as a consequence of
Prop.~\ref{pro:ConvergentDilation}, $\afk$ has a factorizing scaling limit and the representation 
of the symmetry group $\gcal_0$
factorizes too. It is also clear from the proof above and from that of
Thm.~\ref{thm:convergentImpliesUnique}, that $\qafk$ is unitarily
equivalent to $\afk$ through the operator $W$, taken here for $\qpi$ in place
of $\pi_0$. Furthermore, this $W$ also intertwines the dilations in the scaling
limit with those of the underlying theory.

\begin{corollary}\label{cor:dilation}
Under the hypothesis of Thm.~\ref{thm:dilation}, there holds
\begin{equation*}\label{eq:dilationequivalence}
W \qU(\mu) W^* = U(\mu).
\end{equation*}
\end{corollary}

\begin{proof}
It is sufficient to verify the relation on vectors of the form $A \Omega$ with
$A \in \hat\afk(\ocal)$. For such vectors it follows by noting that 
$\uA_\lambda = U(\lambda)AU(\lambda)^*$ is an element of
$\uafk_\mathrm{conv}(\ocal)$, and that $\udelta_\mu(\uA)_\lambda =
U(\lambda)\alpha_\mu(A)U(\lambda)^*$ with $\alpha_\mu(A) \in \hat\afk(\mu \ocal)$.
\end{proof}

For showing the consistency of our definitions, we now prove that the
Haag-Swieca compactness condition at finite scales, together with dilation
covariance, implies our uniform compactness condition of Def.~\ref{def:uniformCompactness}.

\begin{proposition} \label{pro:compactnessFiniteToUniform}
If the dilation covariant local net $\afk$ fulfills the Haag-Swieca compactness
condition, then it also fulfills uniform Haag-Swieca compactness.
\end{proposition} 

\begin{proof}
Let $\ocal$ be fixed, and let $\beta>0$ such that $\Theta^{(\beta,\ocal)}$ is
compact;
\begin{equation}
  \Theta^{(\beta,\ocal)} = \sum_{j=1}^\infty e_j(\cdotarg) f_j
  \quad \text{with } e_j \in \afk(\ocal)\st, \; f_j \in \hcal.
\end{equation}
Taking the normal part, we can in fact arrange that $e_j \in
\afk(\ocal)_\ast$. (See \cite[Lemma~2.2]{BDF:universal_structure} for a similar
argument.) Now define $\ue_j: \uafk(\ocal) \to \zfk(\uafk)$ by
\begin{equation}
  \ue_j(\uA)_\lambda = e_j\big( U(\lambda)\st \cutoff_\lambda \uA_\lambda
  U(\lambda) \big).
\end{equation}
That the image is indeed in $\zfk(\uafk)$, i.e., continuous under $\udelta_\mu$,
is seen as follows. We compute for $\lambda,\mu>0$,
\begin{multline}
  \big| \ue_j(\uA)_{\lambda\mu} - \ue_j(\uA)_\lambda  \big|
  = 
  \big| e_j\Big( U(\lambda\mu)\st (\cutoff\uA)_{\lambda\mu} U(\lambda\mu) -
  U(\lambda)\st (\cutoff\uA)_\lambda U(\lambda) \Big)
  \big|
  \\
  \leq
  \|e_j\| \, \|\udelta_\mu(\cutoff\uA) - \cutoff\uA\| + 
  \|e_j\big( U(\mu)\st \cdotarg U(\mu) \big) 
  - e_j\|\, \|\cutoff\uA\|.
\end{multline}
Now as $\mu \to 0$, the first summand vanishes due to norm continuity of
$\udelta_\mu$ on $\uafk$, and the second due to strong continuity of
$U(\mu)$; both limits are uniform in $\lambda$. Thus
$\udelta_\mu$ acts continuously on $\ue_j(\uA)$.

Further, we define $\uf_j \in \uhcal$ by
\begin{equation}
  \uf_{j \lambda} = U(\lambda) f_j.
\end{equation}
This is indeed an element of $\uhcal$: Namely, given $\epsilon >0$, choose $A
\in\hat\uafk(\ocal)$ for suitable $\ocal$ such that $\|A\Omega-f_j\|<\epsilon$;
here $\hat\afk$ is as in the proof of Thm.~\ref{thm:dilation}.
Then, $\uA_\lambda = U(\lambda)AU(\lambda)\st$ defines an element of $\uafk$,
and $\|\uA\Omega-\uf_j\| \leq \|A\Omega-f_j\|<\epsilon$. Hence $\uf_j$ is
contained in the closure of $\uafk\Omega$.

Now we are in the position to show that $\uTheta^{(\beta,\ocal)} = \sum_j \ue_j
\uf_j$. Let $J\in\nbb$ be fixed. It is straightforward to compute that for
any $\uA \in \uafk(\ocal)$ and $\lambda>0$,
\begin{multline}
  \Big( \uTheta^{(\beta,\ocal)}(\uA) - \sum_{j=1}^{J} \ue_j(\uA)
  \uf_j\Big)_\lambda
  =  U(\lambda) \Big( \Theta^{(\beta,\ocal)}(B_\lambda) -\sum_{j=1}^J
  e_j\big(B_\lambda) f_j \Big),
    \\
    \text{where } B_\lambda = U(\lambda)\st \cutoff_\lambda \uA_\lambda
    U(\lambda).
\end{multline}
Note here that $B_\lambda \in \afk(\ocal)$ for any $\lambda$. This entails
\begin{equation}
  \big\| \uTheta^{(\beta,\ocal)} - \sum_{j=1}^{J} \ue_j(\cdotarg) \uf_j \big\|
  \leq \big\| \Theta^{(\beta,\ocal)} - \sum_{j=1}^{J} e_j(\cdotarg)f_j
  \big\|\,\|\uA\|.
\end{equation} 
The right-hand sides vanishes as $J \to \infty$, as a consequence of the
compactness condition at finite scales. This shows that
$\uTheta^{(\beta,\ocal)}$ is uniformly compact.
\end{proof}

\subsection{The scaling limit of a free field} 

We now show in a simple, concrete example from free field theory that the model
has a convergent scaling limit in the sense of Def.~\ref{def:convLimit}.
Specifically, we consider a real scalar free field of mass $m>0$, in
2+1 or 3+1 space-time dimensions. The algebraic scaling limit of this model is
the massless real scalar field; this was as already discussed in
\cite{BucVer:scaling_examples}, and in parts we rely on the arguments given
there. However, we need to consider several aspects that were not handled in that work, in particular continuity aspects of
Poincaré and dilation transformations. Also, as
mentioned before, in contrast to \cite{BucVer:scaling_examples} we deal with
weakly closed local algebras at fixed scales and in the limit theory. 

We start by recalling, for convenience, the necessary notations and definitions
from~\cite{BucVer:scaling_examples}. We consider the Weyl algebra $\wfk$ over $\dcal(\rbb^s)$, $s=2,3$:
\begin{equation} \label{eqn:WeylAlg}
W(f)W(g) = e^{-\frac{\imath}{2}\sigma(f,g)}W(f+g),\qquad\sigma(f,g) = \im
\int d^{s}x\, f(\xv)g(\xv),
\end{equation}
Then, we define a mass dependent automorphic action of $\poincare$ on $\wfk$
by
\begin{equation}
\alpham_{x,\Lambda}(W(f)) = W(\taum_{x,\Lambda}f),
\end{equation}
where the action $\taum$ of $\poincare$ on $\dcal(\rbb^s)$ is defined by the
following formulas. In those, we write $\tilde f(\pv) = (2\pi)^{-s/2}
\int d\xv\,f(\xv)e^{-i\xv\pv}$ for the Fourier transform of $f$, which we
split into $\tilde f = \tilde f_R + \imath \tilde f_I$, where $f_R = \re
f$ and $f_I = \im f$; also,  $\omega_m(\pv) :=
\sqrt{m^2+|\pv|^2}$.
\begin{align}
(\taum_\xv f)(\yv) &:= f(\yv-\xv),\\
\begin{split} \label{eq:timetrans}
(\taum_t f)\sptilde_R(\pv) &:= \cos(t\omega_m(\pv)) \tilde f_R (\pv) -
\omega_m(\pv)\sin(t\omega_m(\pv))\tilde{f}_I (\pv),\\ 
(\taum_t
f)\sptilde_I(\pv) &:=\cos(t\omega_m(\pv)) \tilde{f}_I(\pv) +
\omega_m(\pv)^{-1}\sin(t\omega_m(\pv))\tilde{f}_R(\pv), \end{split}
\\
\label{eq:lorentz}\begin{split}
(\taum_\Lambda f)\sptilde_R(\pv) &:= \varphi^f_\Lambda(\omega_m(\pv),\pv),\\
(\taum_\Lambda f)\sptilde_I(\pv) &:=
\omega_m(\pv)^{-1}\psi^f_\Lambda(\omega_m(\pv),\pv).
\end{split}
\end{align}
Here the functions $\varphi^f_\Lambda, \psi^f_\Lambda : \rbb^{s+1}\to \cbb$
are defined by
\begin{equation}\label{eq:lorentz2}
\begin{split}
\varphi^f_\Lambda(p) &:=
\frac{1}{2}\Big(\tilde{f}_R(\Lambdav^{-1}\pv)+\tilde{f}_R(\Lambdav^T\pv)\Big)
+\frac{1}{2\imath}\Big( (\Lambda^{-1}p)_0\tilde{f}_I
(\Lambdav^{-1}\pv)-(\Lambda^Tp)_0\tilde{f}_I(\Lambdav^T\pv)\Big),\\
\psi^f_\Lambda(p) &:=
-\frac{1}{2\imath}\Big(\tilde{f}_R(\Lambdav^{-1}\pv)-\tilde{f}_R(\Lambdav^T\pv)\Big)
+\frac{1}{2}\Big((\Lambda^{-1}p)_0\tilde{f}_I
(\Lambdav^{-1}\pv)+(\Lambda^Tp)_0\tilde{f}_I (\Lambdav^T\pv)\Big),
\end{split}
\end{equation}
where we use the notation $\Lambda p = ((\Lambda p)_0, \Lambdav \pv)$ for $\Lambda \in \lcal$, $p \in \rbb^{s+1}$. 
One verifies that all the above expressions are even in
$\omega_m(\pv)$, which, due to the analytic properties of $\tilde f$, 
implies that $\taum_{x,\Lambda}\dcal(\rbb^s) \subset \dcal(\rbb^s)$. 
We also introduce the action $\sigma$ of dilations on $\wfk$ by
\begin{equation}
\sigma_\lambda (W(f)) = W(\delta_\lambda f),
\end{equation}
with
\begin{equation}
(\delta_\lambda f)(\xv) := \lambda^{-(s+1)/{2}}(\re f)(\lambda^{-1}\xv) +
\imath\lambda^{-(s-1)/2}(\im f)(\lambda^{-1}\xv).
\end{equation}
It holds that $\alpham_{\lambda x,\Lambda}\circ \sigma_\lambda =
\sigma_\lambda \circ \alphalm_{x,\Lambda}$. Finally we define the vacuum state of mass $m \geq 0$ on $\wfk$ as
\begin{equation}
\omegam(W(f)) = e^{-\frac{1}{2}\|f\|_m^2},
\end{equation}
where
\begin{equation}\label{eq:norm}
\|f\|_m^2 := \frac{1}{2}\int_{\rbb^s}d\pv \big|
\omega_m(\pv)^{-1/2}\tilde{f}_R(\pv)+\imath\,\omega_m(\pv)^{1/2}\tilde{f}_I(\pv)\big|^2.
\end{equation}
There holds clearly $\omegam \circ \alpham_{x,\Lambda}=\omegam$, $\omegam \circ \sigma_\lambda = \omegalm$.

Proceeding now along the lines of~\cite{BucVer:scaling_examples}, we consider
the GNS representation $(\pi^{(0)}, \hcal^{(0)},\Omega^{(0)})$ of $\wfk$
induced by the massless vacuum state $\omega^{(0)}$. For each $m\geq 0$, we
define a net $\ocal \mapsto \afkm(\ocal)$ of von Neumann algebras on $\hcal^{(0)}$ as
\begin{equation}
\afkm(\Lambda \ocal_B + x) :=
\{\pi^{(0)}\big(\alpham_{x,\Lambda}(W(g))\big)\,:\,\supp g \subset B\}'',
\end{equation}
where $\ocal_B$ is any double cone with base the open ball $B$ in the time
$t=0$ plane. For other open regions we can define the algebras by taking unions,
but this will not be relevant for the following discussion. Due to the
local normality of the different states $\omegam$, $m\geq 0$, with respect to each other~\cite{EckFro:local_equiv}, these nets are isomorphic to the nets generated by the free scalar field of mass $m$ on the
respective Fock spaces. From now on, we will identify elements of $\wfk$ and of
$\afkm$, and therefore we will drop the indication of the representation
$\pi^{(0)}$. We denote the (dilation and Poincar\'e covariant) scaling algebra
associated to $\afkm$ by $\uafkm$. The next lemma generalizes the results
of~\cite[Lemma 3.2]{BucVer:scaling_examples} to the present situation.

\begin{lemma}\label{lem:weyllimit}
Let $a>1$ and $h^D \in \dcal((1/a,a))$, $h^P \in \dcal(\poincare)$, $f \in \dcal(\rbb^s)$, and consider the function $\uW : \rbb_+ \to \afkm$ given by
\begin{equation*}
\uW_\lambda := \int_{\rbb_+\times\poincare}\frac{d\mu}{\mu}\,dx\,d\Lambda\,
h^D(\mu)h^P(x,\Lambda)\alpham_{\mu\lambda x,\Lambda}\circ\sigma_{\mu\lambda}(W(f)),
\end{equation*}
where $d\Lambda$ is the left-invariant Haar measure on the Lorentz group and the
integral is to be understood in the weak sense. Then:
\begin{enumerate}
\localitemlabels
\item \label{Wlocal}there exists a double cone $\ocal$ such that $\uW \in \uafkm(\ocal)$;
\item \label{Wlimit}there holds in the strong operator topology,
\begin{equation*}
\lim_{\lambda \to 0^+}\sigma_\lambda^{-1}(\uW_\lambda) =
\int_{\rbb_+\times\poincare}\frac{d\mu}{\mu}\,dx\,d\Lambda\,
h^D(\mu)h^P(x,\Lambda)\alphaz_{\mu x,\Lambda}\circ\sigma_\mu(W(f)) =: W_0;
\end{equation*}
\item \label{Wdense}the span of the operators $W_0$ of the form above, with
$\uW \in \uafkm(\ocal)$ for fixed $\ocal$, is strongly dense in $\afkm(\ocal)$.
\end{enumerate} 
\end{lemma}

\begin{proof}
Since $\uW$ is the convolution, with respect to the action $(\mu,x,\Lambda)\mapsto \ualpham_{\mu,x,\Lambda}$, of the function $h^D \otimes h^P$ with the bounded function $\lambda \mapsto \sigma_\lambda(W(f))$, and thanks to the support properties of $h^D$, $h^P$ and $f$,~\ref{Wlocal} follows.

In order to prove~\ref{Wlimit}, we start by observing that, for each vector $\chi \in \hcal^{(0)}$,
\begin{multline}
\|\Big(\sigma_\lambda^{-1}(\uW_\lambda)-W_0\Big)\chi\| \leq
\int_{\rbb_+\times\poincare}\frac{d\mu}{\mu}\,dx\,d\Lambda\,
|h^D(\mu)h^P(x,\Lambda)|
\\
\times \big\| \Big( W\big(\delta_\mu \tau^{(\mu\lambda
m)}_{x,\Lambda}f\big)-W\big(\delta_\mu \tauz_{x,\Lambda}f\big) \Big)\chi\big\|.
\end{multline}
Now $f \mapsto W(f)$ is known to be continuous with respect to
$\|\cdotarg\|_0$ on the initial space and the strong operator
topology on the target space \cite[Prop.~5.2.4]{BraRob:qsm2}.
Since the norm $\|\cdot\|_0$ is $\delta_\mu$-invariant, it therefore suffices to
show that for each fixed $(x,\Lambda)\in \poincare$,
\begin{equation}\label{eq:limitnorm}
\lim_{m\to 0^+}\big\| \taum_{x,\Lambda}f - \tauz_{x,\Lambda}f\big\|_0 = 0;
\end{equation}
for \ref{Wlimit} then follows from the dominated convergence theorem. In order
to show Eq.~\eqref{eq:limitnorm}, we introduce the following family of
functions $\fm{m}(\pv)$ of two arguments:
\begin{equation}
\begin{aligned}
   \fcal = \Big\{ f : [0,1] \times \rbb^s \to \cbb \,\Big|\,
         &\fm{m}(\cdotarg) \in \dcal(\rbb^s) \text{ for each fixed }m \in [0,1];
         \\
         & \lim_{m \to 0} \widetilde{\fm{m}}(\pv) = \widetilde{\fm{0}}(\pv)
         \text{ for each fixed }\pv \in\rbb^s;
         \\
         &\exists g \in \scal(\rbb^s) \, \forall m \in [0,1]\, \forall \pv \in
         \rbb^s: \;\;
         | \widetilde{\fm{m}}(\pv) | \leq g(\pv)
          \Big\}.
\end{aligned} 
\end{equation}
It is clear that for $f \in \fcal$, one has $\|\fm{m}-\fm{0}\|_0 \to 0$ as $m
\to 0$ per dominated convergence. Also, each $f \in \dcal(\rbb^s)$, with trivial
dependence on $m$, falls into $\fcal$. So it remains to show that the
(naturally defined) action of $\tausup{\cdotarg}_{x,\Lambda}$ leaves $\fcal$
invariant, where it suffices to check this for a set of generating subgroups. Indeed,
$\tausup{\cdotarg}_{x,\Lambda} \fcal \subset \fcal$ is clear for spatial
translations and rotations. For time translations and boosts, it was already remarked that
$\dcal(\rbb^s)$ is invariant under these at fixed $m$, and pointwise convergence
as $m \to 0$ is clear. Further, from Eqs.~\eqref{eq:timetrans} and \eqref{eq:lorentz2}, one
sees that $\tilde f$ is modified by at most polynomially growing functions,
uniform in $m\leq 1$, hence uniform $\scal$-bounds hold for
$\tausup{m}_{x,\Lambda} \fm{m}$ as well. (Again, it enters here that
all expressions are even in $\omega_m$, for which it is needed that $\tilde f$ is smooth.) This completes the proof of \ref{Wlimit}.
 
Finally,~\ref{Wdense} follows from the observation that as $h_D$ and $h_P$
converge to delta functions,  $W_0$ converges strongly to $W(f)$ thanks to the
strong continuity of the function $(\mu,x,\Lambda) \mapsto \alphaz_{\mu
x,\Lambda}\circ\sigma_\mu(W(f))$; and of course the span of the Weyl operators
with $\supp f \subset\subset\ocal$ is strongly dense in
$\afkm(\ocal)$.
\end{proof}

Using the above lemma, we can prove the following.

\begin{theorem}
The theory of the massive real scalar free field in $s=2,3$ spatial dimensions
has a convergent scaling limit.
\end{theorem}

\begin{proof}
Consider the \cistar-subalgebra $\huafkm(\ocal)$ of $\uafkm(\ocal)$ which is generated by the elements $\uW\in\uafkm(\ocal)$ defined in the previous lemma, and let $\huafkm$ be the corresponding quasi-local algebra. Since $\ualpham_{\mu,x,\Lambda}(\uW)$ is again an element of the
same form, just with shifted function $h^D\otimes h^P$, the algebra $\huafkm$ is
$\ualpham$ invariant. In order to verify that $\lambda \mapsto
\omegam(\uA_\lambda)$ has a limit, as $\lambda \to 0$, for all $\uA \in
\huafkm$, we start by observing that, thanks to
Lemma~\ref{lem:weyllimit}~\ref{Wlimit} and to the fact that $\sigma_\lambda$ is
unitarily implemented on $\hcal^{(0)}$, for each such $\uA$ there exists
$\lim_{\lambda \to 0^+}\sigma_\lambda^{-1}(\uA_\lambda) =: A$ in the strong
operator topology. Then if $\uA \in \huafkm(\ocal)$ there holds the inequality
\begin{equation}
|\omegam(\uA_\lambda) - \omegaz(A)| \leq \| (\omegam-\omegaz)\restrict
\afkz(\lambda \ocal)\| \|\uA\| + |\omegaz(\sigma_\lambda^{-1}(\uA_\lambda)) - \omegaz(A)|.
\end{equation}
Together with the fact that  $\lim_{\lambda\to 0^+} \|
(\omegam-\omegaz)\restrict \afkz(\lambda \ocal)\| = 0$ as a consequence of the
local normality of $\omegam$ with respect to $\omegaz$, this implies that
$\lim_{\lambda \to 0^+}\omegam(\uA_\lambda) = \omegaz(A)$ for all $\uA$ in some
local algebra $\huafkm(\ocal)$. This then extends to all of $\huafkm$ by
density.

It remains to show that for multiplicative limit states, $\pi_0(\huafkm \cap
\uafkm(\ocal))$ is weakly dense in $\afkm_0(\ocal)=\pi_0(\uafkm(\ocal))''$ for
any $\ocal$. To that end, we use similar methods as in Thm.~\ref{thm:dilation}.
With $\ocal$ fixed and $U$ the ultrafilter that underlies the limit state, we define
\begin{equation}
\phi: \pi_0(\uafkm(\ocal)) \to \afkz(\ocal), \quad
  \pi_0(\uA) \mapsto \lim_U \sigma_\lambda^{-1} (\uA_\lambda),
\end{equation}
with the limit understood in the weak operator topology.  
Using methods as in \cite[Sec.~3]{BucVer:scaling_examples}, one can show
that $\phi$ is indeed a well-defined isometric $\st$ homomorphism, which further
satisfies $\omega_0=\omegaz \circ \phi$ on the domain of $\phi$. Hence $\phi$
is given by the adjoint action of a partial isometry, and can be continued by
weak closure to a $\st$ homomorphism $\phi^-: \afkm_0(\ocal) \to
\afkz(\ocal)$. On the other hand, for $\uW \in \huafkm(\ocal)$ as in Lemma~\ref{lem:weyllimit}, one finds
$\phi(\pi_0(\uW))=W_0$ by \ref{Wlimit} of that lemma. However, the double
commutant of those $W_0$ is all of $\afkz(\ocal)$, see \ref{Wdense} of the same
lemma. So $\phi^-$ is in fact an isomorphism; and inverting $\phi^-$, one
obtains the proposed density.
\end{proof}

\subsection{Phase space conditions in the free field} \label{sec:pscExample}

Our aim in this section is to prove the uniform compactness condition of
Sec.~\ref{sec:compactness} in the case of a real scalar free field, again of
mass $m \geq 0$, in $3+1$ or higher dimensions. To that end, we will use a
short-distance expansion of local operators, very similar to the method used in 
the Appendix of \cite{Bos:short_distance_structure}, however in a refined
formulation.

In this section, we will consider a fixed mass $m\geq 0$ throughout, and
therefore we drop the label $(m)$ from the local algebras, the vacuum state, and
the Hilbert space norm for simplicity. We rewrite the Weyl operators of
Eq.~\eqref{eqn:WeylAlg} in terms of the familiar free field $\phi$ and its time
derivative $\partial_0\phi$ in the time-0 plane,
\begin{equation} \label{eqn:WeylAsFields}
  W(f) = \exp \imath \big( \phi(\re f) - \partial_0\phi(\im f) \big), \quad 
  f \in \dcal(\rbb^s).
\end{equation}
Also, we need to introduce some multi-index notation. Given $n \in \nbb_0$, we
consider multi-indexes $\nu = (\nu_1,\ldots,\nu_n) \in (\{0,1\} \times \nbb_0^s)^n$;
that is, each $\nu_j$ has the form
$\nu_j=(\nu_{j0},\nu_{j1},\ldots,\nu_{js})$ with
$\nu_{j0} \in \{0,1\}$, $\nu_{jk} \in \nbb_0$ for $1 \leq k
\leq s$. These indices will be used for labeling derivatives in configuration
space, $\partial^{\nu_j} = \partial_0^{\nu_{j0}}\ldots\partial_s^{\nu_{js}}$.
We denote
\begin{equation}
\nu_j! = \prod_{k=0}^{s} \nu_{jk}!\,, 
\quad
\nu! = \prod_{j=1}^{n} \nu_j!\,,
\quad
 |\nu_j| = \sum_{k=0}^{s} \nu_{jk}\,,
 \quad
 |\nu| = \sum_{j=1}^{n} |\nu_j|.
\end{equation}
Now we can define the following local
fields as quadratic forms on a dense domain.
\begin{equation}
   \phi_{n,\nu} = \;\;\wickprod{ \prod_{j=1}^n \partial^{\nu_j} \phi}(0).
\end{equation}
These will form a basis in the space of local fields at $x=0$. Further, for
given $r>0$, we choose a test function $h\in\dcal(\rbb^s)$
which is equal to 1 for $|\xv|\leq r$; then we set $h_{\nu_j} (\xv) = \prod_{k=1}^s
x_k^{\nu_{jk}}h(\xv)$. This is used to define the following functionals on
$\afk(\ocal_r)$.
\begin{equation} \label{eqn:sigmaNNuDef}
\sigma_{n,\nu} (A) = \frac{\imath^{n} (-1)^{\sum_j \nu_{j0}}}{n! \, \nu!}
  \bighrskp{\Omega}{  [ \partial_0^{(1-\nu_{10})} \phi(h_{\nu_1}),
  [ \ldots [\partial_0^{(1-\nu_{n0})} \phi(h_{\nu_n}) , A] \ldots ] \Omega}.
\end{equation}
One sees that this definition is independent of the choice of $h$. We
can therefore consistently consider $\sigma_{n,\nu}$ as a functional on
$\cup_\ocal \afk(\ocal)$, though its norm may increase as $\ocal$ grows large.
The significance of these functionals becomes clear in the following lemma.

\begin{lemma} \label{lem:taylorExpansionProto}
We have for all Weyl operators $A=W(f)$ with $f \in \dcal(\rbb^s)$,
\begin{equation*}
   A = \sum_{n=0}^\infty \sum_\nu \sigma_{n,\nu} (A) \, \phi_{n,\nu}
\end{equation*}
in the sense of matrix elements between vectors of finite energy and
finite particle number.
\end{lemma}

\begin{proof}
We indicate only briefly how this combinatorial formula can be
obtained; see also \cite[Sec.~7.2]{Bos:Operatorprodukte} and
\cite[Appendix]{Bos:short_distance_structure}. Using Wick ordering, we rewrite
Eq.~\eqref{eqn:WeylAsFields} for the Weyl operators as 
\begin{equation}
   W(f)  = e^{-\|f\|^2/2} \wickprod{\exp \imath \big(\phi(\re
   f)-\partial_0\phi(\im f) \big)}\; = e^{-\|f\|^2/2} \sum_{n=0}^\infty
   \frac{\imath^n}{n!} \, \wickprod{ \big(\phi(\re
   f)-\partial_0\phi(\im f) \big)^n }.
\end{equation}
Now, in each factor of the $n$-fold product $\wickprod{
(\ldots)^n }$, we expand both $\re f$ and $\im f$ into a Taylor series in
momentum space. Note that this is justified, since those functions have
compact support in configuration space, since they are evaluated in scalar products with functions of compact support in momentum space, and since the sum over $n$ is finite in matrix elements. The
Taylor expansion in momentum space corresponds to an expansion in
derivatives of $\delta$-functions in configuration space, and this is what
produces the fields $\phi_{n,\nu}$ localized at 0. We then need to identify the
remaining factors with $\sigma_{n,\nu}(A)$, which is done using the known commutation
relations of $W(f)$ with $\phi$ and $\partial_0\phi$.
\end{proof}

Our main task will now be to extend the above formula to more general states
and more observables, by showing that the sum converges in a suitable norm. To
that end, we need estimates of the fields and functionals involved.

\begin{lemma} \label{lem:fixedScaleEst}
Given $s\geq 2$, $m \geq 0$, and $r_0 > 0$, there exists a constant $c$ such
that the following holds for any $n,\nu$.
\begin{align}
  \label{eqit:betabounds} \tag{a} 
  \|e^{-\beta H} \phi_{n,\nu} \Omega\| &\leq c^n \, (n!)^{1/2} \, \nu! \,
  (\beta/2)^{-|\nu|-n(s-1)/2} &\quad &\text{for any $\beta > 0$},
  \\
  \label{eqit:ebounds} \tag{b}
  \|P(E) \phi_{n,\nu} P(E) \| &\leq c^n E^{|\nu|+n(s-1)/2}
  &\quad &\text{for any $E \geq 1$, provided $s \geq 3$,}
  \\ 
  \label{eqit:rbounds} \tag{c}
  \|\sigma_{n,\nu} \restrict \afk(\ocal_r) \| 
  &\leq c^n \, (n!)^{-1/2} (\nu!)^{-1} \, (3r)^{|\nu|+n(s-1)/2}
  &\quad &\text{for any $r\leq r_0$}.
\end{align}
\end{lemma}
\begin{proof}
One has
\begin{equation}
 \|e^{-\beta H} \phi_{n,\nu} \Omega \| = 
 \big\| \Big(\prod_{j=1}^n
  a\st(e^{-\beta\omega} p^{\nu_j}) \Big)\,\Omega \big\|
  \leq \sqrt{n!} \prod_{j=1}^{n} \|e^{-\beta\omega} p^{\nu_j}\|.
\end{equation}
For the single-particle vectors on the right-hand side, one uses scaling
arguments to obtain the estimate
\begin{equation}
 \|e^{-\beta\omega} p^{\nu_j}\| \leq c_1 \,  \nu_j! \, 
 (\beta/2)^{-|\nu_j|-(s-1)/2} \quad \text{for $\beta > 0$},
\end{equation}
where $c_1$ is a constant (depending on $s$ and $m$). This implies
\eqref{eqit:betabounds}.---For \eqref{eqit:ebounds}, we use energy bounds for
creation operators $a\st(f)$, similar to \cite[Sec.~3.3]{BucPor:phase_space}.
One finds for single particle space functions $f_1,\ldots,f_k$ in the domain of
$\omega_m^{1/2}$,
\begin{equation}
   \| P(E) a\st(\omega_m^{1/2} f_1)\ldots a\st(\omega_m^{1/2} f_k) \|
   \leq E^{k/2} \|Q(E) f_1\| \ldots \|Q(E) f_k\|,
\end{equation}
with $Q(E)$ being the energy projector for energy $E$ in single particle
space. This leads us to 
\begin{equation}
  \|P(E) \phi_{n,\nu} P(E)\| \leq 2^n E^{n/2} \prod_{j=1}^n \|\omega_m^{-1/2}
  p^{\nu_j} \chi_E\|,
\end{equation}
where $p^{\nu_j}=\prod_{k=0}^s p_k^{\nu_{jk}}$, and $\chi_E$ is the
characteristic function of $\omega_m(\pv) \leq E$. For the single-particle
functions, one obtains
\begin{equation}
    \|\omega_m^{-1/2}  p^{\nu_j} \chi_E\| \leq c_2 E^{|\nu_j|+(s-2)/2}
\end{equation}
with a constant $c_2$, which implies \eqref{eqit:ebounds}.

Now consider the functional $\sigma_{n,\nu}$. We choose a test function $h_1 \in
\dcal(\rbb_+)$ such that $h_1(x)\leq 1$ for all $x$, $h_1(x) =1$ on $[0,1]$, and
$h_1(x)=0$ for $x\geq 2$. Then, $h_r(\xv) = h_1(|\xv|/r)$ is a valid choice for
the test function used in the definition of $\sigma_{n,\nu}\restrict\afk(\ocal_r)$,
see Eq.~\eqref{eqn:sigmaNNuDef}. Expressing the fields $\phi$ there in
annihilation and creation operators, and writing each commutator as a sum of
two terms, we obtain
\begin{equation}
  \|\sigma_{n,\nu} \restrict \afk(\ocal_r) \| 
  \leq \frac{4^n}{\sqrt{n!}\, \nu!} \prod_{j=1}^n \| \omega_m^{(1-\nu_{j0})}
    h_{r,\nu_{j}}\|.
\end{equation}
Again, we use scaling arguments for the single-particle space functions and
obtain for $r \leq r_0$,
\begin{equation}
  \| \omega_m^{(1-\nu_{j0})} h_{r,\nu_{j}} \|
  \leq c_3 (3r)^{|\nu_j|+(s-1)/2}
\end{equation}
with a constant $c_3$ that depends on $r_0$. This yields \eqref{eqit:rbounds}.
\end{proof}

We now define the ``scale-covariant" objects that will allow us to expand the
maps $\uTheta^{(\beta,\ocal)}$ in a series. They are constructed of the fields
$\phi_{n,\nu}$ and the functionals $\sigma_{n,\nu}$ by multiplication with
appropriate powers of $\lambda$. We begin with the quantum fields.

\begin{proposition} \label{pro:uphiExist}
For any $n,\nu$ and $\beta >  0$, the function
\begin{equation*}
   \uchi_{n,\nu,\beta} : \lambda \mapsto \lambda^{|\nu|+n(s-1)/2}
   e^{-\beta\lambda H} \phi_{n,\nu} \Omega
\end{equation*}
defines an element of $\uhcal$, with norm estimate
$\|\uchi_{n,\nu,\beta}\| \leq c^n \, (n!)^{1/2} \, \nu! \,
  (\beta/2)^{-|\nu|-n(s-1)/2}$.
\end{proposition}

\begin{proof}
We use techniques from \cite{BDM:scaling_fields}, and adopt the notation
introduced there. In particular, $\uR$ 
denotes the function $\uR_\lambda = (1+\lambda H)^{-1}$, and we write $\lnorm{
\uA }{\ell} = \sup_\lambda \| \uR_\lambda^\ell \uA_\lambda \uR_\lambda^\ell\|$,
where $\uA_\lambda$ may be unbounded quadratic forms. Let $n,\nu$ be fixed in
the following. We set
\begin{equation}
   \uphi_\lambda = \lambda^{|\nu|+n(s-1)/2} \phi_{n,\nu}.
\end{equation}
From Lemma~\ref{lem:fixedScaleEst}, one sees that $\|P(E/\lambda) \uphi_\lambda
P(E/\lambda)\|$ is bounded uniformly in $\lambda$. Hence, applying
\cite[Lemma~2.6]{BDM:scaling_fields}, we obtain $\lnorm{\uphi}{\ell}<\infty$
for sufficiently large $\ell$. Also, the action $g \mapsto \ualpha_g \uphi$ of
the symmetry group on $\uphi$ (which extends canonically from bounded operators
to quadratic forms) is continuous in some $\lnorm{\cdotarg}{\ell}$: This is clear
for translations by the energy-damping factor; for dilations it is immediate
from the definition; and for Lorentz transformations it holds since they act by
a finite-dimensional matrix representation on $\phi_{n,\nu}$. Thus, $\uphi$ is
an element of the space $\uPhi$ defined in
\cite[Eq.~(2.39)]{BDM:scaling_fields}. Moreover, each $\uphi_\lambda$ is
clearly an element of the Fredenhagen-Hertel field content $\PhiFH$. Thus,
\cite[Thm.~3.8]{BDM:scaling_fields} provides us with a sequence $(\uA_n)$ in
$\uafk(\ocal)$, with $\ocal$ a fixed neighborhood of zero, such that
$\lnorm{\uA_n-\uphi}{\ell} \to 0$ as $n \to \infty$. Now since
$\|\exp(-\beta \uH) \uR^{-\ell}\|< \infty$, 
we obtain $\|\exp(-\beta\uH) (\uA_n-\uphi) \Omega \| \to 0$. 
Note here that $\exp(-\beta \uH) \uA \Omega \in \uhcal$ by
Lemma~\ref{lem:fBeta}. Hence $\exp(-\beta \uH) \uphi \Omega =
\uchi_{n,\nu,\beta}$ lies in $\uhcal$, since this space is closed in norm.
The estimate for $\|\uchi_{n,\nu,\beta}\|$ follows directly from Lemma~\ref{lem:fixedScaleEst}\eqref{eqit:betabounds}.
\end{proof}

Next, we rephrase the functionals $\sigma_{n,\nu}$ as maps from the scaling
algebra $\uafk$ to its center.

\begin{lemma} \label{lem:usigmaExist}
For any $n,\nu$, the definition
\begin{equation*}
   (\usigma_{n,\nu} (\uA))_\lambda =  \lambda^{-|\nu|-n(s-1)/2}
    \sigma_{n,\nu}((\cutoff\uA)_\lambda)
\end{equation*}
yields a linear map $\usigma_{n,\nu}:\cup_\ocal \uafk(\ocal) \to
\zfk(\uafk)$, with its norm bounded by 
\begin{equation*}
\|\usigma_{n,\nu} \restrict \uafk(\ocal_r) \|  \leq c^n
(\sqrt{n!} \;\nu!)^{-1} \, (3r)^{|\nu|+n(s-1)/2}
\end{equation*}
for $r \leq r_0$; here $r_0,c$ are the constants of
Lemma~\ref{lem:fixedScaleEst}.
\end{lemma}
\begin{proof}
The norm estimate is a consequence of
Lemma~\ref{lem:fixedScaleEst}\eqref{eqit:rbounds}, where one notes that
$(\cutoff\uA)_\lambda \in \afk(\ocal_{\lambda r})$ for $\lambda \leq 1$, and
$(\cutoff\uA)_\lambda=0$ for $\lambda>1$, so that these operators are always
contained in $\afk(\ocal_{r_0})$. It remains to show 
that $\usigma_{n,\nu}(\uA) \in \zfk(\uafk)$, i.e., that $\mu \mapsto \udelta_\mu (\usigma_{n,\nu}(\uA))$ is continuous. But this follows from continuity of $\mu\mapsto \udelta_\mu \uA$ and the definition of $\usigma_{n,\nu}$.
\end{proof}

We are now in the position to prove that $\sum_{n,\nu}
\usigma_{n,\nu}\uchi_{n,\nu,\beta}$ is a norm convergent expansion of the map
$\uTheta^{(\beta,\ocal)}$.

\begin{theorem} \label{thm:sumConverge}
Let $s \geq 3$. For each $r>0$, there exists $\beta>0$ such that
\begin{equation*}
  \uTheta^{(\beta,\ocal_r)}= \sum_{n=0}^\infty \sum_\nu
  \usigma_{n,\nu}(\cdotarg)
  \uchi_{n,\nu,\beta} \,.
\end{equation*}
\end{theorem}
\begin{proof}
We will show below that the series in the statement converges in norm in the Banach space
$\bcal(\uafk(\ocal_r),\uhcal)$, i.e.\ that
\begin{equation} \label{eqn:convGoal}
\sum_{n=0}^\infty \sum_\nu \| \usigma_{n,\nu} \restrict \uafk(\ocal_r) \| \|
\uchi_{n,\nu,\beta} \| < \infty.
\end{equation}
Once this has been established, the assertion of the theorem is obtained as follows.
From Lemma~\ref{lem:taylorExpansionProto}, we know that
\begin{equation}\label{eqn:serieslambda}
\sum_{n=0}^\infty \sum_\nu \lambda^{-|\nu|-n(s-1)/2}\sigma_{n,\nu}(A)
(\uchi_{n,\nu,\beta})_{\lambda} = e^{-\lambda \beta H}A\Omega 
\end{equation}
at each fixed $\lambda$, 
whenever $A$ is a linear combination of Weyl operators, and when evaluated in scalar
products with vectors from a dense set. But~\eqref{eqn:convGoal}
also shows that the left hand side of~\eqref{eqn:serieslambda} converges 
in $\bcal(\afk(\lambda \ocal_r),\hcal)$, and it is therefore strongly continuous for $A$
in norm bounded parts of $\afk(\lambda\ocal_r)$.
Then by Kaplansky's theorem~\eqref{eqn:serieslambda} holds for any $A \in \afk(\lambda\ocal_r)$
as an equality in $\hcal$. Finally, this entails for all $\uA \in
\uafk(\ocal_r)$ that $\sum_{n,\nu}\usigma_{n,\nu}(\cutoff\uA)_\lambda
(\uchi_{n,\nu,\beta})_\lambda = \uTheta^{(\beta,\ocal_r)}(\uA)_\lambda$ at each fixed $\lambda > 0$, i.e.\ the statement.

We now prove Eq.~\eqref{eqn:convGoal}.
Let $r_0>0$ be fixed in the following, and $r<r_0$. From
Prop.~\ref{pro:uphiExist} and Lemma~\ref{lem:usigmaExist}, we obtain the estimate
\begin{equation}
\begin{aligned}
\| \usigma_{n,\nu}  \restrict \uafk(\ocal_r) \| \| \uchi_{n,\nu,\beta} \|
  &\leq 
  c^{2n}
  \, (6r/\beta)^{|\nu|+n(s-1)/2}
\\&=
    \Big(c^2 (6r/\beta)^{(s-1)/2}\Big)^n \prod_{j=1}^n
  \left((6r/\beta)^{\nu_{j0}} \prod_{k=1}^s 
  (6r/\beta)^{\nu_{jk}}\right).
\end{aligned}
\end{equation}
Factorizing the sum over multi-indexes $\nu$ accordingly, we obtain at fixed
$n$,
\begin{equation} \label{eqn:fixedNEst}
   \sum_\nu \| \usigma_{n,\nu}  \restrict \uafk(\ocal_r) \| \|
   \uchi_{n,\nu,\beta} \| \leq  \Big(2c^2  \frac{(6r/\beta)^{(s-1)/2}}{(1-6r/\beta)^s} \Big)^n,
\end{equation}
where we suppose $6r/\beta < 1$, and where each sum over $\nu_{j0}\in\{0,1\}$
has been estimated by introducing a factor of $2$.
Now if we choose $r/\beta$ small enough, we can certainly achieve that the
expression in Eq.~\eqref{eqn:fixedNEst} is summable over $n$ as a geometric
series, and hence the series in Eq.~\eqref{eqn:convGoal} converges.
\end{proof}

This establishes the phase space condition of Def.~\ref{def:uniformCompactness}
in our context.
\begin{corollary}
The theory of a real scalar free field of mass $m \geq 0$ in $3+1$ or more
space-time dimensions fulfills the uniform Haag-Swieca compactness condition.
\end{corollary}

While our goal was to show that the maps $\uTheta^{(\beta,\ocal)}$
are uniformly compact, it follows from Eq.~\eqref{eqn:convGoal} that they are actually
uniformly nuclear at all scales, or by a slightly modified argument, even
uniformly $p$-nuclear for any $0<p\leq 1$. So we can generalize the somewhat
stronger Buchholz-Wichmann condition \cite{BucWic:causal_independence} to our
context. Several other types of phase space conditions can be derived with
similar methods as in Thm.~\ref{thm:sumConverge} as well. Particularly, one can
show for $s \geq 3$ that the sum $\sum_{n,\nu} \sigma_{n,\nu} \phi_{n,\nu}$
converges in norm under a cutoff in energy $E$ and restriction to a fixed local algebra
$\afk(\ocal_r)$, with estimates uniform in $E \cdot r$ where this product is
small. This implies that Phase Space Condition II of \cite{BDM:scaling_fields},
which guarantees a regular behavior of pointlike fields under scaling, is fulfilled for those
models.

\section{Conclusions} \label{sec:conclusion}

In this paper, we have considered short-distance scaling limits in the
model independent framework of Buchholz and Verch
\cite{BucVer:scaling_algebras}. In order to describe the dilation symmetry that
arises in the limit theory, we passed to a generalized class of limit states,
which includes states invariant under scaling. The essential results of
\cite{BucVer:scaling_algebras} carry over to this generalization, including
the structure of Poincaré symmetries and geometric modular action. 

However, the dilation invariant limit states are not pure. Rather, they can be
decomposed into states of the Buchholz-Verch type, which are pure in two or
more spatial dimensions. This decomposition gives rise to a direct integral
decomposition of the Hilbert space of the limit theory, under which local
observables, Poincaré symmetries, and other relevant objects of the theory
can be decomposed---except for dilations. The dilation symmetry has a more
intricate structure, intertwining the pure components of the limit state.

The situation simplifies if we consider the situation of a ``unique limit'' in
the classification of \cite{BucVer:scaling_algebras}; our condition of a
``factorizing limit'' turned out to be equivalent modulo technicalities.
Under this restriction, the dilation unitaries in the limit are, up to a central
part, decomposable operators. The decomposed components do not necessarily
fulfill a group relation though, but a somewhat weaker cocycle relation. Only
under stronger assumptions (``convergent scaling limits'') we were able to
show that the dilation symmetry factorizes into a tensor product of unitary
group representations.

It is unknown at present which type of models would make the generalized
decomposition formulas necessary. In this paper, we have only considered 
very simple examples, which all turned out to fall into the
more restricted class of convergent scaling limits. However, thinking e.g.~of
infinite tensor products of free fields with increasing masses as suggested in
\cite[Sec.~5]{Buc:phase_space_scaling}, it may well be that some models violate
the condition of uniqueness of the scaling limit, or even exhibit massive
particles in the limit theory. In this case, the direct integral
decomposition would be needed to obtain a reasonable description of dilation
symmetry in the limit, with the symmetry operators intertwining fibers of the
direct integral that correspond to different masses.
 
As a next step in the analysis of dilation symmetries in the limit theory, one
would like to investigate further the \emph{deviation} of the theory at finite
scales from the idealized dilation covariant limit theory; so to speak, the
next-order term in the approximation $\lambda \to 0$. This would be interesting
e.g.~for applications to deep inelastic scattering, which can currently only
be treated in formal perturbation theory. It is expected that the dilation
symmetry in the limit also contains some information about these next-order
terms. Our formalism, however, does at the moment not capture these next-order
approximations, and would need to be generalized considerably.

Further, it would be interesting to see whether the dilation symmetry we
analyzed can be used to obtain restrictions on the type of interaction in the
limit theory, possibly leading to criteria for asymptotic freedom. Here we do
not refer to restrictions on the form of the Lagrangian, a concept that is not
visible in our framework. Rather, we think that dilation symmetries should
manifest themselves in the coefficients of the operator product expansion
\cite{Bos:product_expansions,BDM:scaling_fields,BosFew:qi_ope} or in the general
structure of local observables \cite{BucFre:dilations}.

In this context, it seems worthwhile to investigate simplified
low-dimensional interacting models, such as the 1+1 dimensional massive models
with factorizing S matrix that have recently been rigorously constructed in the algebraic
framework \cite{Lec:factorizing_s}. By abstract arguments, these models possess
a scaling limit theory in our context. Just as in the Schwinger model
\cite{Buc:quarks,BucVer:scaling_examples}, one expects here that even the limit
theory for multiplicative states has a nontrivial center. This may be
seen as a peculiarity of the 1+1 dimensional situation; we have not specifically dealt
with this problem in the present paper. But neglecting these aspects, one would 
expect that the limit theory corresponds to a \emph{massless} (and dilation
covariant) model with factorizing S matrix, although it would probably not have
an interpretation in the usual terms of scattering theory. Such models of
``massless scattering" have indeed been considered in the physics literature; see \cite{FenSal:massless_scattering} for a review. Their mathematical status as quantum field theories, however, remains largely unclear at this time.
Nevertheless, one should be able to treat them with our methods.

In fact, these examples may give a hint to the restrictions on interaction that
the dilation symmetry implies. At least in a certain class of
two-particle S matrices---those which tend to 1 at large momenta---one expects
that the limit theory is \emph{chiral,} i.e., factors into a tensor product of
two models living on the left and right light ray, respectively. On the
other hand, the theories we consider are always \emph{local;} and for chiral
local models, dilation covariance is the essential property that guarantees
\emph{conformal} covariance \cite{GLW:extensions}. So if those models have a
nontrivial scaling limit at all, they underly quite rigid restrictions, since
local conformal chiral nets are---at least partially---classifiable in a
discrete series \cite{KawLon:classification_local}. The detailed investigation of these aspects of factorizing
S matrix models is however the subject of ongoing research, and some surprises
are likely to turn up.

\appendix 

\section{Direct integrals of Hilbert spaces} \label{sec:diApp}

In our investigation, we make use of the concept of direct integrals of Hilbert
spaces, $\hcal = \int_\zcal d\nu(z) \hcal_z$, where the integral is defined on
some measure space $(\zcal,\nu)$. Due to difficult measure-theoretic problems, the
standard literature treats these direct integrals only under
separability assumptions on the Hilbert spaces involved; see
e.g.~\cite{KadRin:algebras2}. These are however not a priori implied in our
analysis; and even where we make such assumptions, we need to apply them with
care. While we can often reasonably assume the ``fiber spaces'' $\hcal_z$ to be separable,
the measure space $\zcal$ will, in our applications, be of a very general
nature, and even $L^2(\zcal,\nu)$ is known to be non-separable in some
situations. The concept of direct integrals can be generalized to that case.
Since however the literature on that topic\footnote{%
In the general case, we largely follow Wils \cite{Wil:direct_integrals_1},
however with some changes in notation. Other, somewhat stronger notions of
direct integrals exist, e.g.~\cite[Ch.~III]{God:repres_unit}, \cite{Seg:decomp_oa}; see
\cite{Mar:hilbert_integral} for a comparison.  
Under separability assumptions (Definition~\ref{def:fundamentalSequence}), all
these notions agree, and we are in the case described in
\cite[Ch.~IV.8]{Tak:TOA1}, \cite[Part~II Ch.~1]{Dix:v_n_algebras}.} 
is somewhat scattered and not easily accessible, we give here a brief review
for the convenience of the reader, restricted to the case that concerns us.

In the following, let $\zcal$ be a compact topological space and $\nu$ a finite
regular Borel measure on $\zcal$. For
each $z \in \zcal$, we consider a Hilbert space $\hcal_z$ with scalar product
$\etskp{\cdotarg}{\cdotarg}_z$ and associated norm $\|\cdotarg\|_z$. Elements
$\chi \in \prod_{z\in\zcal} \hcal_z$ will be called \emph{vector fields} and
alternatively denoted as maps, $z \mapsto \chi(z)$. Direct integrals of this
field of Hilbert spaces $\hcal_z$ over $\zcal$ are not unique, but depend on the
choice of a \emph{fundamental family}.

\begin{definition} \label{def:fundamentalFamily}
A \emph{fundamental family} is a linear subspace $\Gamma \subset
\prod_{z\in\zcal} \hcal_z$ such that for every $\chi \in \Gamma$, the function
$z \mapsto \|\chi(z)\|_z^2$ is $\nu$-integrable. If the same function is
always continuous, we say that $\Gamma$ is a \emph{continuous fundamental
family}.
\end{definition}

The continuity aspect will be discussed further below, for the moment we
focus on measurability. Each fundamental family $\Gamma$ uniquely extends
to a minimal vector space $\bar\Gamma$, with $\Gamma \subset \bar\Gamma \subset \prod_{z\in\zcal} \hcal_z$, which has the following properties. \cite[Corollary~2.3]{Wil:direct_integrals_1}
\begin{enumerate}
\localitemlabels
  \item $z \mapsto \|\chi(z)\|_z^2$ is $\nu$-integrable for all
  $\chi \in \bar\Gamma$.
  \item If for $\chi \in \prod_z \hcal_z$, there exists $\hat\chi \in
  \bar\Gamma$ such that $\chi(z)=\hat\chi(z)$ a.e., then $\chi
  \in\bar\Gamma$.
  \item If $\chi\in\bar\Gamma$ and $f \in L^\infty(\zcal,\nu)$, then $f\chi \in
  \bar\Gamma$, where $(f\chi)(z):=f(z)\chi(z)$.
  \item $\bar\Gamma$ is complete with respect to the seminorm $\|\chi\| =
  (\int_\zcal d\nu(z) \|\chi(z)\|^2)^{1/2}$.
\end{enumerate}
Such $\bar\Gamma$ is called an \emph{integrable family}. It
is obtained from $\Gamma$ by multiplication with $L^\infty$ functions and
closure in $\|\cdotarg\|$. The elements of $\bar\Gamma$ are called
\emph{$\Gamma$-measurable functions;} they are in fact analogues of
square-integrable functions, and the usual measure theoretic results hold 
for them: Egoroff's theorem; the dominated convergence theorem; 
any norm-convergent sequence $(\chi_n)$ in
$\bar\Gamma$ has a subsequence on which $\chi_n(z)$ converges pointwise
a.e.; and if $(\chi_n)$ is a sequence in $\bar\Gamma$ that converges pointwise
a.e., the limit function $\chi$ is in $\bar\Gamma$. (Cf.~Propositions 1.3 and
1.5 of \cite{Mar:hilbert_integral}.) After dividing out vectors of zero norm
(which we do not denote explicitly), $\bar\Gamma$ becomes a Hilbert space, which we call
the \emph{direct integral} of the $\hcal_z$ with respect to $\Gamma$, and
denote it as
\begin{equation}
 \hcal = \int_\zcal^\Gamma d\nu(z) \hcal_z, \quad
 \text{with scalar product } \hrskp{\chi}{\hat \chi} =
 \int_\zcal d\nu(z) \etskp{\chi(z)}{\hat\chi(z)}_z \,.
\end{equation}
Correspondingly, the elements $\chi \in \hcal$ are denoted as $\chi  =
\int_\zcal^\Gamma \chi(z) d\nu(z)$.

We also consider bounded operators between such direct integral spaces. Let
$\hcal_z$, $\hat \hcal_z$ be two fields of Hilbert spaces over the same measure
space $\zcal$, and let $\Gamma,\hat \Gamma$ be associated fundamental
families. We call $B \in \prod_{z \in \zcal} \bcal(\hcal_z,\hat\hcal_z)$ a
\emph{measurable field of operators} if $\esup_z \|B(z)\|
<\infty$, and if for every $\chi \in \bar\Gamma$, the vector field $z \mapsto
B(z)\chi(z)$ is $\hat\Gamma$-measurable, i.e.~an element of
$\bar{\hat\Gamma}$. In fact it suffices to
check the measurability condition on the fundamental family $\Gamma$.

\begin{lemma} \label{lem:measOnGamma}
Let $B \in \prod_{z \in \zcal} \bcal(\hcal_z,\hat\hcal_z)$ such that $\esup_z
\|B(z)\| <\infty$, and such that $z
\mapsto B(z)\chi(z)$ is $\hat\Gamma$-measurable for every $\chi \in \Gamma$.
Then $B$ is a measurable field of operators.
\end{lemma}
\begin{proof}
Evidently, $z \mapsto B(z) \chi(z)$ is also $\hat\Gamma$-measurable if $\chi$ is
taken from $L^\infty (\zcal,\nu) \cdot \Gamma$ or from its linear span. This
span is however dense in $\hcal$. So let $\chi \in\hcal$. There exists a
sequence $\chi_n \in \lspan L^\infty (\zcal,\nu) \cdot \Gamma$ such that
$\chi_n \to \chi$ in norm; by the remarks after
Def.~\ref{def:fundamentalFamily}, we can assume that $\chi_n(z) \to \chi(z)$ a.~e. 
But then $B(z) \chi_n(z) \to B(z)
\chi(z)$ a.~e., due to continuity of each $B(z)$. So $B(z)\chi(z)$ is a
pointwise a.~e.~limit of functions in $\bar{\hat\Gamma}$. This implies $(z
\mapsto B(z) \chi(z)) \in \bar{\hat \Gamma}$, which was to be shown.
\end{proof}

A measurable field of operators $B$ now defines a bounded
linear operator $\hcal \to \hat\hcal$ which we denote as $B =
\int_\zcal^{\Gamma,\hat\Gamma} d\nu(z) B(z)$.  Operators in
$\bcal(\hcal,\hat\hcal)$ of this form are called \emph{decomposable}. Their
decomposition need not be unique however, not even a.~e. If here
$\Gamma=\hat\Gamma$, and if $B(z)=f(z)\idop_{\hcal_z}$ with
a function $f \in L^\infty(Z,\nu)$, then $B$ is called a \emph{diagonalizable} operator. We sometimes write this multiplication operator as $M_f = \int_\zcal^\Gamma d\nu(z) f(z) \idop$.

Let $\afk$ be a \cistar{} algebra, and let for each $z\in\zcal$ a representation
$\pi_z$ of $\afk$ on $\hcal_z$ be given, such that $z \mapsto \pi_z(A)$ is a
measurable field of operators for any $A\in\afk$. Then, $\pi(A) =
\int_\zcal^\Gamma d\nu(z) \pi_z(A)$ defines a new representation $\pi$ of $A$
on $\hcal$, which we formally denote as $\pi = \int_\zcal^\Gamma d \nu(z)
\pi_z$.

In many cases, obtaining useful results regarding decomposable operators
requires additional separability assumptions. The following property
will usually be general enough for us.

\begin{definition} \label{def:fundamentalSequence}
A \emph{fundamental sequence} in $\hcal = \int_\zcal^\Gamma d\nu(z) \hcal_z$ is
a sequence $(\chi_j)_{j\in\nbb}$ in $\bar\Gamma$  
such that for every $z \in \zcal$, the
set $\{\chi_j (z) \,|\, j \in \nbb\}$ is total in $\hcal_z$. If such a fundamental sequence exists for $\hcal$, we say that $\hcal$ is
\emph{uniformly separable}.
\end{definition} 

This more restrictive situation agrees with the setting
of \cite{Tak:TOA1,Dix:v_n_algebras};
cf.~\cite[Prop.~1.13]{Mar:hilbert_integral}.
Note that this property implies that the fiber spaces $\hcal_z$ are separable,
but the integral space $\hcal$ does not need to be separable if $\zcal$ is
sufficiently general. Under the above separability assumption,
additional desirable properties of decomposable operators hold true.

\begin{theorem}
Let $\hcal = \int_\zcal^\Gamma d\nu(z) \hcal_z$ and $\hat\hcal =
\int_\zcal^{\hat\Gamma} d\nu(z) \hat \hcal_z$ both be uniformly separable. Then, for
each decomposable operator $B = \int_\zcal^{\Gamma,\hat\Gamma} d\nu(z) B(z)$,
also $B\st$ is decomposable, with $B\st = \int_\zcal^{\hat\Gamma,\Gamma}
d\nu(z) B(z)\st$. One has $\|B\|=\|B\st\| = \esup_z \|B(z)\|$.
Decompositions of operators are unique in the following sense: If
$\int_\zcal^{\Gamma,\hat\Gamma} d\nu(z) B(z) = \int_\zcal^{\Gamma,\hat\Gamma}
d\nu(z) \hat B(z)$, then $B(z)=\hat B(z)$ for almost every $z$.
\end{theorem}

For the proof methods, see e.g.~\cite[Ch.~III Sec.~13]{God:repres_unit}. Note
that the theorem is false if the separability
assumption is dropped; see Example~7.6 and Remark~7.11
of \cite[Ch.~IV]{Tak:TOA1}. We also obtain an important characterization of
decomposable operators.
\begin{theorem} \label{thm:decompCommute}
Let $\hcal,\hat \hcal$ be uniformly separable. An operator $B \in
\bcal(\hcal,\hat \hcal)$ is decomposable if and only if it commutes with all
diagonalizable operators; i.e.~$M_f B = B \hat M_f$ for all $f \in
L^\infty(Z,\nu)$.
\end{theorem}
A proof can be found in \cite[Ch.~II §2 Sec.~5 Thm.~1]{Dix:v_n_algebras}. 
In particular, if $\hcal=\hat\hcal$, we know that both the decomposable
operators and the diagonalizable operators form \wstar{} algebras, which are
their mutual commutants. Note that the ``if'' part of the theorem is known to be false for sufficiently general direct integrals, violating the separability
assumption \cite{Sch:decomposable_op}.

We now discuss the case of a continuous fundamental family $\Gamma$;
cf.~\cite[Ch.~III Sec.~2]{God:repres_unit}. In this case, we can consider the
space of \emph{$\Gamma$-continuous functions}, denoted $\ccal(\Gamma)$, and
defined as the closed span of $\ccal(\zcal) \cdot \Gamma$ in the supremum norm,
$\|\chi\|_\infty = \sup_{z \in \zcal} \|\chi(z)\|_z$. With this norm,
$\ccal(\Gamma)$ is a Banach space, in fact a Banach module over the commutative
\cistar{} algebra $\ccal(\zcal)$. We have $\ccal(\Gamma) \subset \bar\Gamma$ in a natural
way, and this inclusion is dense, but it is important to note that different
norms are used in these two spaces.

A simple but particularly important example for direct integrals
arises as follows \cite[Ch.~IV.7]{Tak:TOA1}. Let $\qhcal$ be a fixed
Hilbert space, and $\zcal$ a measure space as above. For each $z \in \zcal$,
set $\hcal_z = \qhcal$. Then the set $\Gamma$ of \emph{constant} functions
$\zcal \to \qhcal$ is a continuous fundamental family; and the associated
integrable family $\bar\Gamma$ is precisely the space of all square-integrable, Lusin-measurable
functions $\chi:\zcal \to \qhcal$. We denote the corresponding direct integral
space as $L^2(\zcal,\nu,\qhcal)=\int_\zcal^\oplus d\nu(z) \qhcal$ (with reference to the
``canonical'' fundamental family). This space is
isomorphic to $L^2(\zcal,\nu) \otimes \qhcal$; the canonical
isomorphism, which we do not denote explicitly, maps $f \otimes \chi$ to the
function $z \mapsto f(z)\chi$. In this way, the algebra of diagonal operators
is identified with $L^\infty(\zcal,\nu) \otimes \idop$.

If here $\qhcal$ is separable, then
$L^2(\zcal,\nu,\qhcal)$ is clearly uniformly separable. In this case, a simple
criterion identifies the elements of the integral space: A function 
$\chi:\zcal\to\qhcal$ is Lusin measurable if and only if it is
weakly measurable, i.e. if $z \mapsto \etskp{\chi(z)}{\eta}$ is measurable for
any fixed $\eta \in \qhcal$. Also, the algebra of
decomposable operators is isomorphic to $L^\infty(\zcal,\nu) \bar\otimes
\bcal(\qhcal)$.

\section*{Acknowledgements} 
  The authors are obliged to Laszlo Zsido and Michael Müger for helpful
  discussions. They also profited from financial
  support by the Erwin Schrödinger Institute, Vienna, and from the friendly
  atmosphere there. HB further wishes to thank the II. Institut für
  Theoretische Physik, Hamburg, for their hospitality.

\bibliographystyle{alpha}
\bibliography{qft}

\end{document}

%% file: macros.tex
\usepackage{amsmath,amsfonts}
\usepackage{latexsym}

\def \cbb{\mathbb{C}}

\def \nbb{\mathbb{N}}
\def \rbb{\mathbb{R}}

\def \bcal {\mathcal{B}}
\def \ccal {\mathcal{C}}
\def \dcal {\mathcal{D}}
\def \ecal {\mathcal{E}}
\def \fcal {\mathcal{F}}
\def \gcal {\mathcal{G}}
\def \hcal {\mathcal{H}}

\def \lcal {\mathcal{L}}

\def \ncal {\mathcal{N}}
\def \ocal {\mathcal{O}}
\def \pcal {\mathcal{P}}

\def \rcal {\mathcal{R}}
\def \scal {\mathcal{S}}

\def \vcal {\mathcal{V}}
\def \wcal {\mathcal{W}}

\def \zcal {\mathcal{Z}}

\def \afk  {\mathfrak{A}}
\def \bfk  {\mathfrak{B}}

\def \wfk  {\mathfrak{W}}
\def \zfk  {\mathfrak{Z}}

\def \<  {\langle}
\def \>  {\rangle}

\def \. { \,\! }

\def \cdotarg { \, \cdot \, }

\def \im  {\mathrm{Im}\,}
\def \re  {\mathrm{Re}\,}

\def \idop {\mathbf{1}}

\def \st {^\ast}
\def \isom {\cong}
\def \restrict {\lceil}

\newsymbol \subsetneqq 2324
\newsymbol \subsetneq 2328

\def \supp {\mathrm{supp}\,}

\def \expltext#1 {\\ \text{\footnotesize{ (#1) }}\\}
\def \intercomm#1 {\\ \text{\footnotesize{ (#1) }}\\}
\def \undercomm#1 {\underset{\text{\scriptsize{ (#1) }}}}
\def \overcomm#1 {\overset{\text{\scriptsize{ (#1) }}}}

\newcommand{\etskp}[2]{\langle #1 | #2 \rangle }
\newcommand{\hrskp}[2]{( #1 | #2 ) }
\newcommand{\bighrskp}[2]{\big( #1 \,\big|\, #2 \big) }

\def \boundedops {\bfk(\hcal)}

\def \dim {\mathrm{dim}\,}
\def \ad  {\mathrm{ad}\,}

\def \tracenorm#1 { \| #1 \| _1 }
\def \hsnorm#1 { \| #1 \| _2 }

\def \vectorcomp#1 {
  \left( \begin{array}{c}
  #1
  \end{array}
  \right)  }

\def \pder#1#2 { \frac{ \partial #1 }{ \partial #2 }}

%% file: scaling-macros.tex
\newcommand{\localitemlabels}{
  \renewcommand{\theenumi}{(\roman{enumi})}
  \renewcommand{\labelenumi}{\theenumi}
}

\newcommand{\mkraum}{ \rbb^{s+1} }

\newcommand{\poincare}{\pcal_+^{\uparrow}}

\DeclareMathOperator{\img}{img}
\DeclareMathOperator{\lspan}{span}

\DeclareMathOperator{\clos}{clos}

\DeclareMathOperator{\esup}{ess\,sup}

\newcommand{\uomega}{\underline{\omega}}
\newcommand{\uA}{\underline{A}}

\newcommand{\uB}{\underline{B}}
\newcommand{\uC}{\underline{C}}

\newcommand{\uH}{\underline{H}}
\newcommand{\ue}{\underline{e}}
\newcommand{\uf}{\underline{\smash f}}

\newcommand{\uR}{\underline{R}}
\newcommand{\uafk}{\underline{\afk}}
\newcommand{\ubfk}{\underline{\bfk}}
\newcommand{\ualpha}{\underline{\alpha}}
\newcommand{\udelta}{\underline{\delta}}
\newcommand{\usigma}{\underline{\sigma}}

\newcommand{\uphi}{\underline{\phi}}
\newcommand{\uPhi}{\underline{\Phi}}
\newcommand{\ugcal}{\smash{\underline{\gcal}}}
\newcommand{\uhcal}{\underline{\hcal}}
\newcommand{\uchi}{\underline{\smash\chi}}
\newcommand{\uTheta}{\underline{\Theta}}
\newcommand{\cutoff}{\underline{C}^{<}}

\def \pv {\mathbf{p}}

\def \xv {\mathbf{x}}
\def \yv {\mathbf{y}}

\newcommand{\lnorm}[2]{  \| #1 \|^{( #2 )}  }

\def \PhiFH { \Phi_{ \mathrm{FH} }  }

\newcommand{\wickprod}[1] {  : \!\! #1 \!\! : }

\def \cdotarg { \, \cdot \, }

\def \im  {\mathrm{Im}\,}
\def \re  {\mathrm{Re}\,}

\def \< {\langle}
\def \> {\rangle}

\def \st {^\ast}
\def \restrict {\lceil}

\newcommand{\mtxe}[3]{ \hrskp{#1}{#2 | #3} }

\def \boundedops {\bcal(\hcal)}

\def \dim {\mathrm{dim}\,}
\def \ad  {\mathrm{ad}\,}

\def \tracenorm#1 { \| #1 \| _1 }
\def \hsnorm#1 { \| #1 \| _2 }

\newcommand{\weakstar}{weak$\ast$ }
\newcommand{\cistar}{C$^\ast$}
\newcommand{\wstar}{W$^\ast$}

\newcommand{\mean}{\mathsf{m}}

\newcommand{\alpham}{\alpha^{(m)}}
\newcommand{\alphaz}{\alpha^{(0)}}
\newcommand{\alphalm}{\alpha^{(\lambda m)}}
\newcommand{\tausup}[1]{\tau^{(#1)}}
\newcommand{\taum}{\tausup{m}}
\newcommand{\tauz}{\tausup{0}}
\newcommand{\omegam}{\omega^{(m)}}
\newcommand{\omegaz}{\omega^{(0)}}
\newcommand{\omegalm}{\omega^{(\lambda m)}}
\newcommand{\fm}[1]{f^{(#1)}}

\newcommand{\Lambdav}{\boldsymbol{\Lambda}}

\newcommand{\afkm}{\afk^{(m)}}
\newcommand{\afkz}{\afk^{(0)}}
\newcommand{\uafkm}{\underline{\afk}^{(m)}}
\newcommand{\huafkm}{\smash{\underline{\afk}}_{\mathrm{conv}}^{(m)}}
\newcommand{\uW}{\underline{W}}
\newcommand{\ualpham}{\underline{\alpha}^{(m)}}

\newcommand{\qafk}{\afk_\mathsf{u}}
\newcommand{\qhcal}{\hcal_\mathsf{u}}
\newcommand{\qOmega}{\Omega_\mathsf{u}}
\newcommand{\qU}{U_\mathsf{u}}

\newcommand{\quomega}{\uomega_\mathsf{u}}
\newcommand{\qpi}{\pi_\mathsf{u}}